\documentclass[11pt,reqno]{amsart}

\usepackage[round]{natbib}
\bibpunct{(}{)}{,}{}{}{), (}
\setcitestyle{notesep={, },round,aysep={,},yysep={,}}
\usepackage{setspace,booktabs}

\usepackage{geometry}
\usepackage{amsmath, amsthm}
\usepackage{amsfonts,amssymb,dsfont}
\usepackage{nicefrac,mathrsfs}
\usepackage{bm}                                 
\usepackage[backgroundcolor=white,bordercolor=orange]{todonotes}

\newcommand{\ccF}{{\mathscr F}}

\newcommand{\Ind}{{\mathds 1}}
\newcommand{\ind}[1]{\Ind_{\{#1\}}}

\newcommand{\R}{\mathbb{R}}

\newcommand{\C}{\mathbb{C}}

\newtheorem{theorem}{Theorem}[section]
\newtheorem{corollary}[theorem]{Corollary}      
\newtheorem{lemma}[theorem]{Lemma}              
\newtheorem{proposition}[theorem]{Proposition}  

\theoremstyle{definition}

\newtheorem{remark}{Remark}[section]

\newtheorem{assumption}{Assumption}[section]

\RequirePackage{colortbl}
\definecolor{tscolor}{rgb}{1.0,0.6,0.0}
\usepackage[ulem=normalem,final]{changes}      
\definechangesauthor[name = Thorsten Schmidt, color=tscolor]{TS}  

\usepackage[colorlinks,urlcolor=red,citecolor=blue,linkcolor=red]{hyperref}
 
\usepackage{amsaddr}

\usepackage{enumerate}

\renewcommand{\i}{{\rm i}}

\usepackage{doc}

\makeatletter
\def\@setthanks{\vspace{-\baselineskip}\def\thanks##1{\@par##1\@addpunct.}\thankses}
\makeatother

\begin{document}

\title[]{{ \Large Variable Annuities in a L\'evy-based hybrid model with surrender risk}}

\author[]{Laura Ballotta$^{(\lowercase{a})}$, Ernst Eberlein$^{(\lowercase{b}),(\lowercase{c}),(\lowercase{d})}$, Thorsten Schmidt$^{(\lowercase{b}),(\lowercase{c}),(\lowercase{d}),\ast}$, Raghid Zeineddine$^{(\lowercase{d})}$}

\thanks{$^{(\lowercase{a})}$ Faculty of Finance, Cass Business School, City, University of London, UK}

\thanks{$^{(\lowercase{b})}$ Freiburg Institute for Advanced Studies (FRIAS), Germany}
\thanks{$^{(\lowercase{c})}$ Department of Mathematical Stochastics, University of Freiburg, Germany}

\thanks{$^{(\lowercase{d})}$ University of Strasbourg Institute for Advanced Study (USIAS), France}

\thanks{$^{\ast}$Corresponding Author: Thorsten Schmidt; email: thorsten.schmidt@stochastik.uni-freiburg.de}

\date{\today}

\maketitle

\begin{abstract}
This paper proposes a market consistent valuation framework for variable annuities with guaranteed minimum accumulation benefit, death benefit and surrender benefit features. The setup is based on a hybrid model for the financial market and uses time-inhomogeneous L\'evy processes as risk drivers. Further, we allow for dependence between financial and surrender risks. Our model leads to explicit analytical formulas for the quantities of interest, and practical and efficient numerical procedures for the evaluation of these formulas. We illustrate the tractability of this approach by means of a detailed sensitivity analysis of the price of the variable annuity and its components with respect to the model parameters. The results highlight the role played by the surrender behaviour and the importance of its appropriate modelling.
\end{abstract}


\section{Introduction}

Variables Annuities (VAs) are unit-linked investment policies providing a post-retirement income, which is generated by the returns on a suitably managed financial portfolio. Various guarantees are applied with the aim of providing protection of the policyholders' saving accounts. VAs are popular insurance products in the US, Japan, the UK, and are increasingly present in the other European markets as well. According to the Life Insurance and Market Research Association (LIMRA) Secure Retirement Institute and the Insured Retirement Institute (IRI), VAs sales for 2018 in the US were more than \$100 billion - a 2\% increase compared to 2017.

Common types of guarantees offered by variable annuity contracts are the so-called Guaranteed Minimum Accumulation Benefit (GMAB), the Guaranteed Minimum Death Benefit (henceforth DB) - which applies in case of early death - the Guaranteed Minimum Income Benefit and the Guaranteed Minimum Withdrawal Benefit. The former two offer protection during the accumulation period, i.e. up to the expiration of the contract, whilst the latter two provide payouts after expiration, during the so-called `annuitization' period. For an extensive overview and classification of these products we refer to \cite{Bacinello2011} and references therein.

Due to the construction of these contracts, the underwriting insurance companies are exposed to financial and mortality risk, as well as surrender risk originated by the policyholder behaviour. Indeed, the option to leave the contract prior to maturity is a common additional feature of insurance contracts, which might cause significant cash outflows for the insurer, and negatively impact on the market growth of VAs (LIMRA Secure Retirement Institute).

The study of the pricing of life insurance contracts in presence of financial risk has been pioneered by \cite{Brennan1976} and \cite{Boyle1977}; an extensive literature has developed since, from the seminal contributions of \cite{albizzatigeman,Bacinello1996} and \cite{grosenjorgensen}, to the more recent ones of \cite{Bacinello2011,Deelstra2013,Giacinto2014} and \cite{Gudkov2018}, to mention a few.
These contributions distinguish themselves in terms of the specific product under consideration, although they are all based on diffusion-driven market models. Extensions to financial dynamics driven by L\'evy processes are considered in \cite{Ballotta2005,Ballotta2009}, although in these papers the focus is primarily on equity markets as the rate of interest is assumed constant and the possibility of surrender is not considered.

In light of the above, the aim of this paper is to provide a realistic framework for the modelling of these risks and the market consistent pricing of variable annuity contracts. We focus specifically on the pricing of the GMAB, the DB and the quantification of the surrender risk by means of the so-called Surrender Benefit (SB), contributing to the current state of the literature in a number of ways. Firstly, contrary to the literature mentioned above, we propose a general joint model for financial and insurance risks, which is driven by time-inhomogeneous L\'evy processes. Our choice is motivated by the increased distributional flexibility offered by these stochastic processes in portraying observed market trends \citep[see, for example,][for an extensive empirical analysis of equity markets]{Eberlein1995}. 

In more details, for the financial market we adopt the hybrid construction of \cite{ER} in which the stochastic dependence between interest rate markets and stock markets is taken into account explicitly. This feature is an important aspect in this context given the typically long maturity of insurance contracts. 

Further, we subdivide insurance risk into surrender risk and mortality risk. For the surrender risk, we follow market practice by considering two components: one capturing the baseline surrender behaviour due to non-economic factors and personal contingencies, and the second one which instead is responsive to changes in the conditions of the financial markets \citep[see][for example]{Kolkiewicz2006,Courtois2011,reacfin16}. This component in particular includes a function of the spread between the rate of the contract and the rate offered on the market for equivalent products, and therefore incorporates the stochastic inputs from the financial market model. The function of choice is designed as to suitably accommodate for surrender triggering arguments based on the moneyness hypothesis \citep[see][for example]{knolleretal}, the interest rate hypothesis and the emergency fund hypothesis \citep[see, for example,][and references therein]{NOLTE201712}. 

For mortality risk, we adopt an extended Gompertz-Makeham model with stochastic mortality improvement ratio as in \cite{KrayzlerZagstBrunner12} and \cite{Escobarea2016}; however we make the natural assumption of stochastic independence between the demographic and the financial risks.

Secondly, through the proposed general framework we obtain closed analytical formulas (up to a multidimensional integral) for the price of the guarantees and the SB. The dimensionality of these integrals is dictated by the frequency with which the policyholder is allowed to terminate the contract.

Finally, we also develop a practical and efficient numerical scheme for the evaluation of these formulas by means of Monte Carlo integration with importance sampling, and we illustrate its applicability by performing a sensitivity analysis of the contract's value with respect to the model parameters. The results from this analysis underline the importance of a correct calibration of the model to the observed surrender rates.

The paper is organized as follows. In section \ref{sec:market} we introduce the model for the financial market; the contract features of the VAs are introduced in section \ref{sec:VA} together with the additional modelling assumptions regarding mortality and surrender risk. The closed analytical pricing formulas are derived in section \ref{sec:pricing}; the corresponding numerical scheme is offered in section \ref{sec:numerics}. We present the results of the sensitivity analysis in section \ref{sec:sensitivity}, whilst section \ref{Conclusion} concludes.

Additional material, including some of the proofs, is provided in the appendix.

\section{The interest rate and equity market}\label{sec:market}
The aim of this section is to introduce a joint model for the interest rate and the equity markets used to develop the pricing framework for variable annuities. Specifically, we follow \cite{ER} and build this joint model on time-inhomogenous L\'evy processes. Thus we first offer some necessary preliminary results with the pricing application in view.

\subsection{Time-inhomogenous L\'evy process and pricing.}
Given a finite time horizon $T^*>0$, consider a stochastic basis $(\Omega, \ccF, \mathbb{F}, Q)$ with a filtration $\mathbb{F}=(\ccF_t)_{t\in [0,T^*]}$ satisfying the usual conditions.  Due to our focus on pricing, the probability measure $Q$ represents  a risk-neutral martingale measure.
Let $L^1$ and $L^2$ be two independent time-inhomogeneous L\'evy processes, i.e. continuous-time processes with independent increments, and with characteristic function
\begin{align}\label{eq:FourierTransformL}
  E_Q\Big[e^{iuL^j_t}\Big]= \exp \bigg(\int_0^t\Big(iub_s^j -\frac12c_s^ju^2 
  + \int_{\mathbb{R}}\big(e^{iux}-1-iux\ind{|x|\leq 1}\big)F^j_s(dx)\Big)ds\bigg)
\end{align}
for $j=1,2$; the local characteristic $(b_s^j, c_s^j, F_s^j)_{s\in [0, T^*]}$ satisfies
the integrability condition
\[
\int_0^{T^*}\bigg(|b_s^j| + c_s^j+ \int_{\mathbb{R}}(\min{\{|x|^2, 1\}})F_s^j(dx)\bigg)ds < \infty.
\]

We note that a financial market built on (exponential) L\'evy processes is in general incomplete, and consequently the risk neutral martingale measure is not unique. Standard practice in this case is to fix the pricing measure via calibration, using market quotes for traded derivatives contracts written on the quantities of interest, i.e. bonds and stocks in the specific case of our application. As model calibration is not in the scope of this paper, we refer to \cite{ER} for a detailed illustration of the calibration procedure for the adopted joint model. Nevertheless, for any risk neutral martingale measure to be well defined, we require that the exponential moments of a certain order exist. To this purpose, the following assumption holds throughout the rest of this paper.

\begin{assumption}[Exponential moments]
\label{ass:EM}
For $j=1,2$, there exist positive constants $M_j$ and $\epsilon_j$ such that for each $u \in [-(1+\epsilon_j)M_j, (1+\epsilon_j)M_j]$ 
\[
\int_0^{T^*}\int_{\{|x|> 1\}}e^{ux}\, F^j_s(dx)ds < \infty .
\]
\end{assumption}

All standard processes used in mathematical finance such as hyperbolic, normal Inverse Gaussian, Variance Gamma and CGMY processes satisfy the above condition.
Assumption \ref{ass:EM} implies the existence of the first moment of the process; this allows us to rewrite \eqref{eq:FourierTransformL} as 
\begin{align}
  E_Q\Big[e^{iuL^j_t}\Big]= \exp \bigg(\int_0^t\Big(iub_s^j -\frac12c_s^ju^2 
  + \int_{\mathbb{R}}\big(e^{iux}-1-iux\big)F^j_s(dx)\Big)ds\bigg). \notag
\end{align}
In this setting, we define the cumulant function  of $L^j$
\[
\theta^j_s(z) = b_s^jz + \frac12c_s^jz^2 + \int_{\mathbb{R}}(e^{zx} -1 -zx)F_s^j(dx),
\]
for any $z\in \mathbb{C}$ such that $Re(z) \in [-(1+\epsilon_j)M_j, (1+\epsilon_j)M_j]$. 
Then, $E_Q[\exp(zL^j_t)]< \infty$ and 
\[
	E_Q\big[\exp(zL^j_t)\big] = \exp\bigg(\int_0^t \theta^j_s(z) ds \bigg).
\]
Further, let $f:\mathbb{R}^+ \to \mathbb{C}$ be a continuous function with $|Re(f)| \leq M_j$, then
\begin{equation}\label{cumulant}
	E_Q\bigg[\exp\bigg(\int_0^tf(s)dL^j_s\bigg)\bigg] = \exp\bigg(\int_0^t \theta^j_s(f(s)) ds \bigg),
\end{equation}
where the integrals are defined component-wise for the real and imaginary part \citep[see][for full details]{EberleinRaible99}.

\subsection{The fixed income market.}
For the modelling of the fixed income market we follow the approach in \cite{EJR} \citep[see also][]{EberleinRaible99}, so that the starting point is the definition of the dynamics of the instantaneous forward rates $(f(t,T)_{0 \le t \le T \le T^*})$.  Let us assume that 
  \begin{equation}
    f(t,T) = f(0,T) + \int_0^t \alpha(s,T)ds - \int_0^t \sigma_1(s,T)dL^1_s, \quad 0 \le t \le T \le T^*, \label{forward rate}
  \end{equation}
for a deterministic and bounded function $f(0,T)$. 
The drift function $\alpha(\cdot)$ and the volatility function $\sigma_1(\cdot)$ are assumed to satisfy the usual conditions of measurability and boundedness \cite[see][(2.5)]{EJR}. 
The price  of a zero coupon bond at time $t$ with maturity $T\ge t$ is 
\[
   B(t,T)= \exp\bigg(-\int_t^T f(t,s)ds\bigg).
\]
Let us denote
\begin{align*}
A(s,T)&:= \int_s^T \alpha(s,u)du,  \quad  \Sigma(s,T):= \int_s^T \sigma_1(s,u)du.
\end{align*}
From Fubini's theorem and equation \eqref{forward rate} it follows that the dynamics of the bond price is
\[
B(t,T)= B(0,T)\exp\bigg(\int_0^t(f(s,s) -A(s,T))ds + \int_0^t \Sigma(s,T)dL^1_s\bigg),
\]

We remind that the short rate $r_t$ is implicitly given by the forward rate dynamics in equation \eqref{forward rate} by setting $T=t$. 
Finally, we assume that 
\begin{equation}
\Sigma(s,T) \leq \frac{M_1}{3}, \label{AssM}
\end{equation}
where $M_1$ is the constant from Assumption \ref{ass:EM}; this guarantees that the exponential of the stochastic integral has finite expectation. 

For the market to be arbitrage free, we require that $(B_t^{-1}B(t,T))_{0 \le t \le T}$, with $B_t = \exp\big(\int_0^tr(s)ds\big)$, is a martingale; it follows from \eqref{cumulant} that  $Q$  is  a martingale measure if
\begin{equation}\label{dc}
   A(s,T)= \theta^1_s(\Sigma(s,T)), \: \: \: 0 \le s \le T.
\end{equation}
In the following we will always assume that the drift condition \eqref{dc} holds. 

\subsection{The stock market}
For the modelling of the equity market, we consider the case of a single asset, be it a single stock or a stock index. More general settings can be obtained in a straightforward manner. 

It is well known from a number of empirical studies that equity and fixed income markets influence each other; this interaction is of particular importance in the context of long-dated insurance contracts.  Thus, following \cite{ER} we choose an approach which allows for stochastic dependence between the two markets. Consequently, we model the price process of the asset as
  \begin{equation}\label{eq:stock}
      S_t=S_0\exp\Big(\int_0^t r(s)ds + \int_0^t \sigma_2(s)dL^2_s + \int_0^t\beta(s)dL^1_s -\omega(t)\Big). 
  \end{equation}

In this hybrid approach, the driver of the interest rate dynamics affects explicitly the stock price in reason of $\beta(\cdot)$ through the term $\int_0^t\beta(s)dL^1_s$.  Further dependence is originated endogenously by the (integrated) short rate, i.e. the classical risk neutral term $\int_0^t r(s)ds$. $\sigma_2(\cdot)$ is a positive function and denotes the volatility of the stock price. Both, $\sigma_2(\cdot)$ and $\beta(\cdot)$, could be chosen as random processes, but having numerical aspects in mind, in the following we consider deterministic functions $\sigma_2(\cdot)$ and $\beta(\cdot)$. 
To ensure the existence of exponential moments we require 
\begin{eqnarray} \label{assM12}
    \sigma_2(s) \leq \frac{M_2}{2}, & &
    |\beta(s)| \leq \frac{M_1}{3},  
\end{eqnarray}
with $M_1, M_2$ the constants from Assumption \ref{ass:EM}. The drift term $\omega(t) $ in \eqref{eq:stock} is chosen such that the discounted stock price $(B_t^{-1}S_t)_{t \in [0,T^*]}$ is a $Q$-martingale. Equation \eqref{cumulant} and the independence of the two driving processes $L^1$ and $L^2$ imply that
  \begin{align}\label{dcII}
  \omega(t)= \int_0^t [\theta^2_s(\sigma_2(s)) +\theta^1_s(\beta(s))]ds.
  \end{align}
Hence, under \eqref{dc} and \eqref{dcII} the joint market model for bonds and security $S$ is free of arbitrage.

\section{The variable annuity contract}\label{sec:VA}

A variable annuity (VA) is an insurance contract which gives the holder a variety of benefits depending on the notional $I$ in exchange for an initial premium which we will determine in the following. This premium is paid once at inception of the VA contract. The maturity $T$ of the contract is 
assumed to satisfy $0<T\le T^*$. In the specification considered here, the VA includes three features: a Guaranteed Minimum Accumulation Benefit (GMAB), a Surrender Benefit (SB), and  a Death Benefit (DB).
 
In details, at maturity $T$ the GMAB pays to the policyholder 
\[ \max( I S_T, G(T)), \]
with $G(T) = I \exp(\delta T)$, $\delta>0$. In other words, the GMAB offers the best of the investment of an amount $I$ in either the asset $S$ or in a risk-free account with guaranteed rate $\delta$. Note that, in order to simplify the notation, in the following we assume that the price process of the asset is normalized so that $ S_0 = 1. $
This payoff, however, can only be claimed if the policyholder is still alive at time $T$ and did not exercise the surrender option before. 
  
In case of early surrender, 
the right of refund is restricted to the current value of the fund reduced by a compulsory surrender penalty. 
Let the penalty $P:[0,T] \to (0,1]$ be an increasing function with $P(T)=1$, and define ${\bf t}:=(t_0,t_1, \ldots, t_K)^\top$ with $0=t_0<t_1< \ldots <t_K< T$. We assume that premature surrender is possible at any time point $t_i\in {\bf t}$ with $i = 1, \ldots,K-1$, in which case the policyholder would receive the amount 
\[    IS_{t_i}P(t_i), \]
 Small values of the penalty $P$ at early dates would allow the insurer to recover any expenses related to the writing of the contract even in case of early surrender. The penalty could also serve to the insurer to hedge against a significant rise in the fund value.

Finally, in case of death before maturity $T$, the death benefit pays (to the beneficiaries) 
\[ \max(IS_{\bar{t}_i}, G(\bar{t}_i)),\]  
      for $\bar t_i\in \bar{\bf t}$, $i = 1,\ldots,N$, and   
      $\bar{\bf t}:=(\bar t_0,\bar{t}_1, \ldots, \bar{t}_N)^\top$ with $0=\bar t_0 < \bar{t}_1< \ldots < \bar{t}_N=T$.   The time points $\bar{t_i}$ denote the frequency with which mortality is monitored by the insurer over the lifetime of the contract.

In principle, the time scales ${\bf t}$ and $\bar {\bf t}$ could be arbitrary, but in typical cases they are not: we assume that ${\bf t}\subset \bar {\bf t}$ in the sense that any surrender time $t_i$ is contained in $\{\bar t_1, \dots,\bar t_N\}$. This is a very natural assumption since death of the policyholder might be monitored by the insurer at the end of every month or every quarter, whereas  surrender of the contract might be allowed only at the end of each year during the life of the contract, or at the policy anniversary.

The description of the contract benefits given above highlights the need for an accurate financial model, but also an appropriate modelling for mortality risk and surrender risk. This is offered in the rest of this section, which we conclude with the discussion of the market consistent valuation of the VA.
 
 \subsection{Mortality model}\label{sec:mortality}
 For the modelling of mortality risk,  we follow standard literature in the field and adopt a stochastic intensity-based approach. Pioneered by \cite{Milevsky2001}, further developed by \cite{Dahl2004, Dahl2006}, and more recently generalized by \cite{Li2011}, this framework builds on a given initial curve for mortality rates by superimposing a stochastic process capturing random improvements and fluctuations.

In details. Let  $\tau^m(x)$ be a random time capturing the remaining lifetime of a $x$ years old individual. The corresponding survival probability with respect to the given probability measure is 
\[Q\left(\tau^m(x)>t\right)=\mathbb{E}_{Q}\left(e^{-\int_{0}^{t}\lambda_{u}^{m}(x+u)du}\right),\]
where $\lambda_{t}^m(x+t)$, $t>0$, is the stochastic mortality intensity for an individual aged $x+t$ at time $t$. This intensity process is modelled as 
\begin{align}\label{eq:lambdam}
\lambda^m_t(x+t)= \lambda^{m,0}(x+t)  \xi_t,   
\end{align}
for an initial curve of the mortality intensity, $\lambda^{m,0}$, and a strictly positive process $\xi_t$ such that $\xi_{0}=1$, capturing the mortality improvements from time 0 to time $t$ for a person aged $x+t$. Finally, the mortality intensity satisfies the property that $\lambda^m_0(x)= \lambda^{m,0}(x)$ (see \citealp{Dahl2006} and references therein for full details).

A popular choice for the initial mortality curve is represented by the Gompertz-Makeham model \citep[][for example]{Dahl2004,Dahl2006}. An alternative specification corresponding to the structure for the UK standard tables for annuitant and pension population is adopted in \cite{Ballotta2006195}. Numerous choices for the process of the mortality improvement ratio, $\xi_t$, have been put forward as well: \cite{Dahl2004,BIFFIS2005443,Dahl2006} for example focus on affine diffusion specifications with particular emphasis on time-inhomogeneous CIR processes; \cite{Ballotta2006195} instead extend the generalized linear model by superimposing a standard mean reverting Ornstein-Uhlenbeck process and suitably accommodating longevity effects.

In light of the above, in the following we choose the standard  Gompertz-Makeham model for the initial curve, so that
\[\lambda^{m,0}(x+t)=\frac{1}{b}e^{\left(\frac{x+t-z}{b}\right)},\] 
and an extended  Ornstein-Uhlenbeck with mean reversion level $e^{-\lambda t}$ for the mortality improvement ratio as in \cite{KrayzlerZagstBrunner12}, i.e.
\[ d\xi(t)=\kappa(\exp(-\lambda t) - \xi(t))dt + \sigma dW_t. \]
We assume that $W$ is a Brownian motion independent of $(L^1, L^2)$, $z$, $\kappa$ and $\sigma$ are non-negative, $b$ is positive and $\lambda \in \mathbb{R}$. This implies that the mortality intensity is independent of the financial market.

Let us denote by $\mathbf{F}^{L^1,L^2} = (\ccF^{L^1,L^2}_t)_{0 \le t \le T^*}$
the filtration generated by $L^1$ and $L^2$. Then, for any set $A\in\ccF^{L^1,L^2}_t$, it follows that
\begin{align}
E_{ Q}\Big[  \ind{ \tau^m(x) >t } \Ind_A \Big] = Q(A) E_{ Q}\Big[  e^{-\int_0^t \lambda^m_u(x+u) du} \Big]. \label{mortality indep}
\end{align}

A few considerations are in order. Firstly, as $\xi$ is modelled as a Gaussian Ornstein-Uhlenbeck process, the mortality intensity can become negative with positive probability. However, this probability is negligible (less than $10^{-5}$) as shown in Appendix A.2. in \cite{Escobarea2016}. For the occurring difficulties with negative values in the intensity, see for example \cite{BieleckiRutkowski2002}, Lemma 9.1.4 and the related remarks.
Secondly, the above assumption of independence between the mortality risk driver, $W$, and the financial risk drivers, $L^1$ and $L^2$, is plausible from the perspective of the insurer: precise modelling would be highly client-specific, and would require information rarely accessible by  insurance companies. Moreover, the independence assumption implies a high degree of tractability which is important for complex products such as VAs considered here. 

Finally, as a consequence of the independence between demographic and financial risks, the computation of the survival probability is carried out under the risk neutral measure $Q$. For specific considerations about the interaction between mortality risk and market risk, we refer the interested reader to \cite{Dahl2006} and references therein.

\subsection{Surrender model}
  
Surrender risk is notoriously difficult to assess and model due to the nature of the decisions leading to it. Policyholders can indeed surrender both because of rising alternative financial opportunities, and apparently non rational (in the financial sense) behaviours due to personal considerations and contingencies. Nevertheless, surrender represents one of the main risks faced by life insurance companies due to the liquidity issues it can generate, with potential loss of market share (see, for example, \citealp{Loisel2011}, and references therein).

A common market and academic practice to surrender modelling \citep[see][for example]{Kolkiewicz2006,Courtois2011,reacfin16} is to consider two components: a deterministic baseline hazard rate function capturing lapses\footnote{`Lapse' was originally used to denote termination of an insurance policy and loss of coverage because the policyholder had failed to pay premia, whilst `surrender' denotes termination accompanied by the payout of a surrender benefit. Nowadays `lapse' often denotes both situations.} due to non-economic factors, and a stochastic process representing additional shocks to the baseline due to changes in the market. 

This random component is usually linked to the spread between the return offered by the contract and the one offered on the market for equivalent products. Indeed, dependence of the surrender on interest rates is relatively intuitive: higher interest rates are strong incentives for the policyholder to switch to higher yield investments, whilst very low interest rates - such as the ones currently observed in all major economies - could represent advantageous opportunities for refinancing. The random component should also be linked to the change in value of the underlying asset, as it directly impacts on the amount received in case of surrender.

Thus, following this line of reasoning, let $\tau^s$ denote the random time of the policyholder decision to surrender. As specified above, surrender is allowed at time points $t_i$, $i=1,\dots,K-1$; by convention  $\{\tau^s = \infty\}$ corresponds to no surrender. Let $\lambda^s$ denote the corresponding intensity of surrender; then $\lambda^s(t)=0$ for $t\in [0, t_1) \cup [t_K, T]$. Further, we model the baseline surrender behaviour due to non-economic factors and personal contingencies by a non-negative constant $C$. 

Let $D(t)$ be the process driving the dynamic lapse component; consistently with the considerations offered above, we build this process on the spread between the return offered by the surrender benefit (net of any penalty charge) plus the market rates at which this amount can be invested, and the total yield of the policy represented by the value at maturity of the guaranteed amount. Thus, let $Y_t= \log S_t$, and $p(t)=- \log P(t)$. Then
\begin{equation}
D(t)= Y_{t} -p(t) + \int_{t}^Tf(t,s)ds - \delta T, \quad 0 \leq  t \leq T. \label{Dt}
\end{equation}
The overall surrender intensity consequently is defined as 
\begin{equation}
\lambda^s(t)= \beta D^2(t_{i}) + C, \quad t_{i} \le t < t_{i+1}, \label{lambda}
\end{equation}
so that it is piecewise constant on the interval $[t_{i},t_{i+1})$ for $i=1, \dots, K-1$. 

The non-negative constant $\beta$ captures the dependence between the surrender intensity and the market and is a measure of the investors' rationality (in the pure economic sense), and their response to personal financial motivations. Equation \eqref{lambda} uses the square of the spread function $D(t)$ in order to capture both situations of favourable market conditions offering more remunerative investment opportunities, and financial market turmoils in which policyholders might lack sufficient resources to finance their expenses (emergency fund hypothesis).

Although in spirit similar to others in the literature (see for example \citealp{Courtois2011} and \citealp{Escobarea2016}), our construction distinguishes itself also for the non-Gaussian dynamics of the underpinning risk drivers $L^{1}$ and $L^{2}$.

The resulting probability of no surrender is given by 
\begin{align} \label{def:taus}
Q(\tau^s \geq t_i |\ccF_{t_i}^{L^1,L^2}) = \exp\Big( -\int_{0}^{t_{i}} \lambda^s_u du \Big) ,
\end{align} 
for all $1 \le  i \le K-1$, and 
\begin{eqnarray}
Q(\tau^s \geq t |\ccF_{t}^{L^1,L^2})&=& Q(\tau^s = \infty |\ccF_{t}^{L^1,L^2})\notag\\ &=& \exp\Big( -\int_{0}^{t_K} \lambda^s_u du \Big),\label{def2:taus}
\end{eqnarray} 
for $t_{K-1} < t$.
We note that the last integral equals $\exp(-\int_{0}^{T} \lambda^s_u du)$. The set-up chosen here can be obtained in a doubly-stochastic model or in a model where immersion holds, see \cite{AksamitJeanblanc} for a comprehensive treatment in this regard. 
An alternative form of the intensity is investigated in the Appendix \ref{sec:altsurrender}.

Finally, we observe that the intensity functions $\lambda^s$ and $\lambda^m$ are independent due to the assumed independence between demographic and financial risks.

\subsection{The price of the variable annuity}
  Using the notation introduced above, we can now formulate the actual cash-flows associated with the considered variable annuity. Firstly, recall that the {\rm GMAB} provides a payoff  only at the maturity time $T$ if the policyholder is alive (i.e., $\{\tau^m(x) > T\}$) and if there was no surrender until this time  (i.e., $\{\tau^s > T\}$). 
  Consequently, the associated cash-flow at maturity $T$ is 
  \begin{equation}\label{GMAB}
  {\rm GMAB}( T)=  \ind{\tau^m(x) >T} \ind{\tau^s >T}  \max(I S_T, G(T)).
  \end{equation}
  Secondly, the surrender option can be exercised only once if the policyholder is still alive (i.e., $\{\tau^s < \tau^m(x)\}$). Therefore, should surrender occur,  the surrender benefit at time $t_i$ pays 
  \begin{equation}\label{SB}
  {\rm SB}(t_i)=  \ind{\tau^s =t_i} \ind{\tau^s < \tau^m(x)} IS_{t_i}P(t_i),
  \end{equation}
  where $t_i$ is one of the possible premature surrender dates. Finally, the death benefit provides a payoff only  in case of no early surrender, and is quantified as 
  \begin{equation}\label{DB}
  {\rm DB}(\bar{t}_i)=  \ind{\bar{t}_{i-1} \leq \tau^m(x) < \bar{t}_i} \ind{\tau^m(x) < \tau^s} \max(I S(\bar{t}_i), G(\bar{t}_i)) ,
  \end{equation}
  where $\bar{t}_i$ is one of the possible payoff dates of the death benefit.

The price  ${\rm P}^{\rm VA}$ of the variable annuity at time $t=0$ is equal to the sum of the prices of its constituents, i.e.
  \begin{align}\label{price:VA}
  {\rm P}^{\rm VA}= {\rm P}^{\rm GMAB} + {\rm P}^{\rm SB} +  {\rm P}^{\rm DB},
  \end{align}
  with, by standard risk neutral valuation argument
  \begin{equation*}
  {\rm P}^{\rm GMAB} = E_{Q}\bigg[e^{-\int_0^T r(u)du} \,  {\rm GMAB}(T)\bigg], 
  \end{equation*}
  \begin{equation*}
  {\rm P}^{\rm SB} =\sum_{i=1}^{K-1}  E_{Q}\bigg[   e^{-\int_0^{t_{i}} r(u)du} {\rm SB}(t_{i})\bigg],
  \end{equation*}
  \begin{equation*}
  {\rm P}^{\rm DB} = \sum_{i=1}^N  E_{Q}\bigg[e^{-\int_0^{\bar t_i} r(u)du} {\rm DB}(\bar t_i)\bigg].
  \end{equation*}

  Tractable pricing formulae for these expressions are provided in the following section. 

  \section{Market-consistent valuation}\label{sec:pricing}
In this section we derive analytical expressions for the components of the variable annuity contract discussed above in the model setup provided in section \ref{sec:market} in presence of mortality and surrender risk. Useful results and representations needed in the following are provided in the  Appendix \ref{app:prelim}.

  \subsection{Guaranteed minimum accumulation benefit}
  The independence between $\tau^m(x)$ and the financial market and result \eqref{def2:taus} imply that
  \begin{align*}
  {\rm P}^{\rm GMAB}&= E_Q\bigg[ e^{-\int_0^T r(u)du}\ind{\tau^m(x) >T} \ind{\tau^s >T} \; \max(I S_T, G(T))\bigg]\\
  &= Q(\tau^m(x) >T) E_Q\bigg[ e^{-\int_0^T r(u)du}\max(I S_T, G(T)) E_Q\bigg(\ind{\tau^s >T}|\ccF^{L^1, L^2}_{T} \bigg) \bigg]\\
  \end{align*}
  and therefore
  \begin{align*}
  {\rm P}^{\rm GMAB}&=  Q(\tau^m(x) >T) E_Q\bigg[ e^{-\int_0^T r(u)du}e^{-\int_0^{t_K}\lambda^s(u)du} \; \max(I S_T, G(T))\bigg].
  \end{align*}
  We introduce the $T$-forward measure $Q^T$ defined as 
  \begin{equation}
  \frac{dQ^T}{dQ}= \frac{1}{B(0,T)B(T)}. \label{forward measure}
  \end{equation}
  Denoting the expectation with respect to $Q^T$ by $E^T$, we obtain 
 \begin{eqnarray*}
  &&{\rm P}^{\rm GMAB}= Q(\tau^m(x) >T)B(0,T) E^T\bigg[ e^{-\int_0^{t_K}\lambda^s(u)du} \; \max(I S_T, G(T))\bigg].\\
  \end{eqnarray*} 
  Observe that 
  \begin{align}\label{eq: max IST and GT} 
  	\max(IS_T, G(T)) = G(T)\bigg[1+ \Big(\frac{IS_T}{G(T)}-1\Big)^+ \bigg],
  \end{align}
  consequently 
 \begin{eqnarray}
  \frac{{\rm P}^{\rm GMAB}}{Q(\tau^m(x) >T)B(0,T)G(T)}&=&  
  E^T\bigg[ e^{-\int_0^{t_K}\lambda^s(u)du}\bigg]
  + E^T\bigg[ e^{-\int_0^{t_K}\lambda^s(u)du} \; \Big(\frac{IS_T}{G(T)}-1\Big)^+ \bigg] \notag\\[2mm]
    &=:& A_1 + A_2,\label{A1A2}
  \end{eqnarray} 
 with obvious definitions in the last line. Thus, the term $A_1$ captures the cost faced by the insurance company of no early surrender; the term $A_2$, instead, is the cost of the option embedded in the GMAB, conditional on no early surrender. 
 
 Let
     
 \begin{align} \label{def:wk}
 	w_l & := \int_0^{t_{l}}A(s,T)ds +\int_0^T f(0,s)ds  - \delta T  -\omega(t_{l}) -p(t_{l})
 	\intertext{ for $l=1,\dots,K-1$ and }
 	w_K & := \int_0^{T}A(s,T)ds +\int_0^Tf(0,s)ds  - \delta T  -\omega(T). \label{def:wk2}
 \end{align}
 
  Further, define $\Delta t_l:=t_l -t_{l-1}$, as well as $R:=(0,\ldots,0,r)\in \mathbb{R}^K$, with $1<r<2$, and let for all $0 \le s \le T$,  and $u \in \R^{K-1}$, $v \in \C^K$
  \begin{eqnarray*}
  D(u,T) &:=& \exp\Big(\i\sum_{l=1}^{K-1} u_{l}w_l \Big), \notag\\
  \tilde D(v,T) &:=& D(v_1,\dots,v_{K-1},T) \exp\Big(\i v_{K}w_K\Big), \notag
  \end{eqnarray*}
  \begin{eqnarray}
  E(s,u,T)&:=& \Sigma(s,T)+ \i(\beta(s) - \Sigma(s,T)) \sum_{l=1}^{K-1}   u_{l} \ind{0\leq s \leq t_{l}}\label{E},\notag\\
  \tilde E(s,v,T)&:=& E(s,v_1,\dots,v_{K-1},T) + \i  (\beta(s) - \Sigma(s,T))v_K,
  \end{eqnarray}
  \begin{eqnarray*}
   F(s,u)&:=& \i \sigma_2(s)\sum_{l=1}^{K-1}  u_{l} \ind{0\leq s \leq t_{l}}, \notag\\
   \tilde F(s,v) &:=& F(s,v_1,\dots,v_{K-1}) + \i  \sigma_2(s)v_K , \notag
   \end{eqnarray*}
   \begin{eqnarray*}
   M(u,T) &:=& D(u,T)\; e^{ \int_0^T\theta^1_s(E(s,u,T))ds + \int_0^T\theta^2_s(F(s,u))ds }
    \; \prod_{l=2}^K \sqrt{\frac{\pi}{ \beta \Delta t_{l}}}e^{-u_{l-1}^2/(4 \beta \Delta t_{l})} , \notag \\
   N(v,T) &:=& \tilde D(v-\i R,T) \; e^{ \int_0^T\theta^1_s(\tilde E(s,v-\i R,T))ds + \int_0^T\theta^2_s(\tilde{F}(s,v-\i R))ds } \notag \\
   &&  \times \frac{1}{(\i v_{K} +r-1)(\i v_{K}+r)} 
\prod_{l=2}^K \sqrt{\frac{\pi}{\beta \Delta t_l}}e^{-v_{l-1}^2/(4 \beta \Delta t_l)} , \text{ for } v\in\mathbb{R}^K. \notag
  \end{eqnarray*}
  
The value of the GMAB is then given in the following.   
  \begin{theorem} \label{prop:A1}
  The price ${\rm P}^{\rm GMAB}$ is given by
  \begin{align*}
  	{\rm P}^{\rm GMAB} = Q(\tau^m(x) >T)B(0,T)G(T) (A_1 + A_2)
  \end{align*}
  with
  \begin{align*}
     A_1 &=\frac{e^{-C(t_K -t_1)}}{(2\pi)^{K-1 }}e^{-\int_0^T A(s,T)ds} 
     \int_{\mathbb{R}^{K-1}}M(u,T)du, \\
     A_2 &= \frac{e^{-C(t_K -t_1)}}{(2\pi)^K}e^{-\int_0^T A(s,T)ds} 
     \int_{\mathbb{R}^{K}} N(u,T) du,
   \end{align*}
   with $A(s,T)$ as in \eqref{dc}.
  \end{theorem}

  \begin{proof}
  We start by calculating $A_1$ and $A_2$, then at the end of the proof we present an explicit formula for $Q(\tau^m(x) >T)$. By the definition of $\lambda^s$ in \eqref{lambda}, we have
\begin{eqnarray*}
 A_1 &=& e^{-C(t_K -t_1)} E^T\bigg[\prod_{i=2}^K e^{-\beta \Delta t_i D(t_{i-1})^2}\bigg]\\ 
 &=& e^{-C(t_K -t_1)} E^T\bigg[f(D(t_1), \ldots, D(t_{K-1}))\bigg],
 \end{eqnarray*}
where 
  $f(x_1, \ldots, x_{K-1})
              :=\prod_{l=2}^{K} \big( e^{-\beta \Delta t_{l} x_{l-1}^2}\big)
               =\prod_{l=2}^{K} f_{l}(x_{l-1}).$ 
For a generic function $f$ we denote by $\hat f$ its Fourier transform. Then, for any $y\in \mathbb{C}$, 
\begin{eqnarray}
\hat{f_l}(y)&=& \int_{\mathbb{R}}e^{\i yt}e^{-\beta \Delta t_l t^2}dt 
= \sqrt{\frac{\pi}{\beta \Delta t_l}}\exp\Big( -y^2/(4 \beta \Delta t_l)\Big). \label{Fourier 1}
\end{eqnarray}
This implies that
\begin{equation}\label{eq:fhat}
 \hat{f}(y_1, \ldots, y_{K-1})=  \prod_{l=2}^K \sqrt{\frac{\pi}{\beta \Delta t_{l}}}e^{-y_{l-1}^2/(4 \beta \Delta t_{l})},
\end{equation}
and we observe that $\hat{f}\in L^1(\mathbb{R}^{K-1})$. By Theorem 3.2 in \cite{EP}
  \begin{equation}
    E^T\bigg[f(D(t_1), \ldots, D(t_{K-1}))\bigg] = \frac{1}{(2\pi)^{K-1}}\int_{\mathbb{R}^{K-1}}\tilde M(\i u)\hat{f}(-u)du, \label{Price-Fourier}
  \end{equation}
  where for any $u=(u_1, \ldots, u_{K-1})$, $\tilde M(\i u)$ is defined as follows
  \begin{align*}
    \tilde M(\i u)&:=E^T\bigg[e^{\i u_1 D(t_1) + \ldots + \i u_{K-1}D(t_{K-1})}\bigg].
  \end{align*}
  Using representation \eqref{D:appendix} of $D$ given in the appendix together with equations \eqref{forward rate} and \eqref{Dt}, we obtain
  \begin{align*}
    \tilde M(\i u) &= \exp\bigg[\i\sum_{l=1}^{K-1} u_{l}\Big(-p(t_{l}) - \delta T+ \int_0^T f(0,s)ds + \int_0^{t_{l}}A(s,T)ds -\omega(t_{l})\Big)\bigg]\\
        & \quad \times E^T\bigg[\exp\Big(\i\sum_{l=1}^{K-1} u_l \Big(\int_0^{t_{l}} \sigma_2(s)dL^2_s + \int_0^{t_{l}}(\beta(s) - \Sigma(s,T))dL^1_s\Big)\Big)\bigg].
  \end{align*}
  Moreover, in virtue of representation \eqref{Bank ac}, the density of $Q^T$ given in \eqref{forward measure} can be written as
  \begin{align*} 
    \frac{dQ^T}{dQ}= \exp\Big( -\int_0^T A(s,T)ds + \int_0^T \Sigma(s,T)dL^1_s \Big).
  \end{align*}
  Consequently
  \begin{align*}
  \lefteqn{ E^T\bigg[\exp\Big( \i\sum_{l=1}^{K-1} (\int_0^{t_{l}} u_{l}\sigma_2(s)dL^2_s + \int_0^{t_{l}}u_{l}(\beta(s) - \Sigma(s,T))dL^1_s)\Big)\bigg]} \hspace{1cm}\\
  &= E_Q\bigg[\exp\Big( \int_0^T\Sigma(s,T))dL^1_s - \int_0^TA(s,T)ds \Big) \\
  & \qquad \quad \times \exp\Big( \i\sum_{l=1}^{K-1} \Big(\int_0^{t_{l}} u_{l}\sigma_2(s)dL^2_s + \int_0^{t_{l}}u_{l}(\beta(s) - \Sigma(s,T))dL^1_s\Big)\Big)\bigg]\\
  &= e^{-\int_0^T A(s,T)ds}E_Q\bigg[\exp\Big( \int_0^T E(s,u,T)dL^1_s + \int_0^T F(s,u)dL^2_s\Big)\bigg]\\
  &= \exp\Big( -\int_0^T A(s,T)ds + \int_0^T\theta^1_s(E(s,u,T))ds + \int_0^T\theta^2_s(F(s,u))ds \Big),
  \end{align*}
  with $E$ and $F$ as in \eqref{E}. The last equality follows from the independence of $L^1$ and $L^2$ and equation \eqref{cumulant}. Therefore, with $D(u,T)$ defined in  \eqref{E} we have
  \begin{eqnarray*}
  \tilde M(\i u)&=& D(u,T)  \exp\Big(-\int_0^T A(s,T)ds + \int_0^T\theta^1_s(E(s,u,T))ds + \int_0^T\theta^2_s(F(s,u))ds \Big)
  \end{eqnarray*}
  and the representation for $A_1$ follows from \eqref{Price-Fourier}.

We continue by calculating 
    $$ A_2 =E^T\bigg[ e^{-\int_0^{t_K}\lambda^s(u)du}  \Big(\frac{IS_T}{G(T)}-1\Big)^+ \bigg]. $$
With the  definitions  $S_T=\exp(Y_T)$,  $G(T) = I \exp(\delta T)$, and \eqref{D:appendix}, we obtain that
\begin{eqnarray*}
 \frac{IS_T}{G(T)} 
 &=& \exp(D_T ).
 \end{eqnarray*}
Therefore, by the same arguments that have been used in the calculation of $A_1$, we can prove that 

\begin{align*}
 A_2 &= e^{-C(t_K -t_1)} E^T\bigg[f(D(t_1), \ldots, D(t_{K-1}))\Big(e^{D(T)}-1\Big)^+\bigg]\\
     &= e^{-C(t_K -t_1)} E^T\bigg[F(D(t_1), \ldots, D(t_{K-1}), D(T))\bigg],
\end{align*}
for
$$ F(x_1, \ldots, x_K):=f(x_1,\dots,x_{K-1})  (e^{x_K }-1)^+ .$$ 
To ensure integrability, we define $g(x_1,\ldots, x_K):=F(x_1, \ldots, x_K) e^{-rx_K}$, with $1<r<2$. Further, let
$$ g_K (x_K) :=  (e^{x_K }-1)^+ e^{-r x_K}. $$
Then, $g_K \in L^1(\mathbb{R})$, such that $g\in L^1(\mathbb{R}^{K})$.
Moreover, elementary integration shows that
for all $ y \in \mathbb{R}$
$$ \hat{g}_K(y)= \frac{1}{(\i y -r+1)(\i y-r)}. $$ 
Observe that $|\hat{g}_K(y)|_{\mathbb{C}}= (((1-r)^2+ y^2)(r^2+y^2))^{-1/2}$, thus,
 $\hat{g}_K\in L^1(\mathbb{R})$. Therefore, combining the last result with \eqref{eq:fhat}, we deduce that $\hat{g}\in L^1(\mathbb{R}^{K})$, and 
 \begin{equation}
   \hat{g}(y_1, \ldots, y_K)=  \frac{1}{(\i y_{K} -r+1)(\i y_{K}-r)}\prod_{l=2}^K \sqrt{\frac{\pi}{\beta \Delta t_l}}e^{-y_{l-1}^2/(4 \beta \Delta t_l)}. \label{Fourier 2}
 \end{equation}
 Since $g, \hat{g}\in L^1(\mathbb{R}^{K})$, we can apply Theorem 3.2 in \cite{EP} and obtain 
 \begin{equation}
E^T\bigg[F(D(t_1), \ldots, D(t_{K-1}), D(T))\bigg] = \frac{1}{(2\pi)^{K}}\int_{\mathbb{R}^{K}}\tilde N(R+\i u)\hat{F}(\i R-u)du, \label{Price-Fourier'}
\end{equation}
with $R:=(0, \ldots, 0, r)\in\mathbb{R}^K$, $1<r<2$, and  $\tilde N(R+\i u)$ defined as 
  \begin{align*}
   \tilde N(R+\i u)&:=E^T\bigg[e^{\i u_1 D(t_1) + \ldots + \i u_{K-1}D(t_{K-1})+ (\i u_K +r)D(T)}\bigg].
  \end{align*}
As above, using the notation from \eqref{E}, we derive
\begin{align*}
\tilde N(R+\i u)
  &= \tilde D(u -\i R ,T)
       e^{-\int_0^T A(s,T)ds}  \\
  & \quad \times E_Q\bigg[\exp\Big( \int_0^{T} \tilde E(s,u-\i R,T) dL^1_s + \int_0^{T} \tilde F(s,u-\i R) dL^2_s \Big)\bigg].
  \end{align*}
 Observe that as $1<r <2$, and in virtue of \eqref{AssM} and \eqref{assM12},  $r \sigma_2(s)\leq M_2$ and $|r\beta(s) +(1-r)\Sigma(s,T)| \leq (2r-1)\frac{M_1}{3} \leq M_1$. Thus, the above expectation exists.
 Using independence of $L^1$ and $L^2$ and \eqref{cumulant}, we obtain that
 \begin{eqnarray}
 \tilde N(R+\i u)
  &=& \tilde D(u-\i R  ,T)  e^{-\int_0^T A(s,T)ds}\notag \\
  && \times \exp\Big( \int_0^T \theta^1_s(\tilde{E}(s,u-\i R,T)) ds + \int_0^T \theta^2_s(\tilde{F}(s,u-\i R))    ds \Big). \label{M'}
 \end{eqnarray}
 On the other hand, observe that for any $u\in \mathbb{R}^{K}$, 
 $$\hat{g}(u)=\int_{\mathbb{R}^{K}}e^{\i \langle u, x \rangle } e^{- \langle R , x \rangle} F(x)dx = \hat{F}(u+\i R). $$ Thus, we deduce that
 \begin{eqnarray}
 \hat{F}(\i R - u) =\hat{g}(-u) 
 &=& \frac{1}{(\i u_{K} +r-1)(\i u_{K}+r)}\prod_{l=2}^K \sqrt{\frac{\pi}{\beta \Delta t_l}}e^{-u_{l-1}^2/(4 \beta \Delta t_l)}. \label{Fourier 3}
 \end{eqnarray}
 Plugging \eqref{M'} and \eqref{Fourier 3} into \eqref{Price-Fourier'}, the claim follows.

 It remains to compute $Q(\tau^m(x) > T)$. Given the setup in section \ref{sec:mortality}, It\^{o}'s lemma and Fubini's theorem (see \citealp{Escobarea2016}, as well for an alternative argument) imply that for any $0\leq t \leq T$ 
\begin{equation}
   Q(\tau^m(x) > t) = \exp\Big(A_x(t) + B_x(t)\lambda_0^m(x)\Big),  \label{Q^m}
\end{equation}
with 
\begin{eqnarray*}
A_x(t)&:=& \frac{c_1\exp(c_2t)}{c_3(c_2+c_3)}[1-\exp(-(c_2+c_3)t)]  \\
&&+\frac14\Big(\frac{c_4}{c_3}\Big)^2 \frac{\exp(2c_5t)}{c_5}[1- \exp(-2c_5t)]\\
&& - \Big(\frac{c_4}{c_3}\Big)^2 \frac{\exp(2c_5t)}{2c_5 +c_3}[1- \exp(-(2c_5+c_3)t)] \\
&&- \frac{c_1\exp(c_2 t)}{c_2c_3}[1-\exp(-c_2 t)]\\
&&+\frac14\Big(\frac{c_4}{c_3}\Big)^2 \frac{\exp(2c_5t)}{c_3+c_5}[1- \exp(-2(c_3+c_5)t)],\\ 
B_x(t)&:=& \frac{1}{c_3}[\exp(-c_3 t) -1],
\end{eqnarray*}
and
\begin{equation*}
 c_1:=  \frac{\kappa}{b}\exp\Big(\frac{x-z}{b}\Big), c_2:=\frac{1}{b} - \lambda, c_3:= \kappa - \frac{1}{b}, c_4:= \frac{\sigma}{b}\exp\Big(\frac{x-z}{b}\Big), c_5:=\frac{1}{b},  
\end{equation*}
with  $ \kappa,  b, z, \lambda$ and $\sigma$ as in section \ref{sec:mortality}.
  \end{proof} 

\subsection{Death benefit}
In this section we compute the value of the death benefit previously defined as
\begin{align*}
      {\rm P}^{\rm DB} &= \sum_{i=1}^N  E_{Q}\bigg[e^{-\int_0^{\bar t_i} r(u)du} {\rm DB}(\bar t_i)\bigg].
  \end{align*}
Recall the definition of $w_l$ from Equations \eqref{def:wk} and \eqref{def:wk2}. Further, for $i\in \{1, \ldots, N\}$, we define
\begin{equation}
	w_{\bar t_i}:=   \int_0^{\bar t_i}A(s,\bar t_i)ds +\int_0^{\bar t_i}f(0,s)ds- \delta \bar t_i -\omega(\bar t_i)  -p(\bar t_i). \label{def:wk3}
 \end{equation}
  We recall that $\Delta t_l:=t_l -t_{l-1}$. For all $j\in \{1,\ldots, K-1\}$, let $R:=(0,\ldots,0,r)\in \mathbb{R}^{j+1}$ with $1<r<2$, $u \in \R^{j}$, and $v \in \C^{j+1}$.
For all $0 \le s \le T$, for all $j\in \{1,\ldots, K-2\}$ and all $i\in \{1, \ldots, N\}$ such that $t_j < \bar t_i \le t_{j+1}$, and for $j=K-1$ and all $i$ such that    $t_{K-1} < \bar t_i \leq T$, we define

\begin{eqnarray*}
D^{j,i}(u,T) &:=& \exp\bigg( \i\sum_{l=1}^{j} u_{l}w_l    \bigg),\notag\\
\tilde D^{j,i}(v,T) &:=& D^{j,i}(v_1,\dots,v_{j},T) \exp\Big(\i v_{j+1}w_{\bar t_i}\Big), 
\end{eqnarray*}
\begin{eqnarray*}
E_{j,i}(s,u,T)&:=& \Sigma(s,\bar t_i)+ \i(\beta(s) - \Sigma(s,T))\sum_{l=1}^{j} u_{l} \ind{0\leq s \leq t_{l}} 
	  ,\notag\\
	 \tilde E_{j,i}(s,v,T)&:=& E_{j,i}(s,v_1,\dots,v_{j},T) + \i  (\beta(s) - \Sigma(s,\bar t_i))v_{j+1},
	 \end{eqnarray*}
\begin{eqnarray}
F_{j}(s,u)&:=& \i\sigma_2(s)\sum_{l=1}^{j} u_{l} \ind{0\leq s \leq t_{l}}, \label{DB notation}\\
 \tilde F_j(s,v)&:=& F_j(s,v_1,\dots,v_{j}) + \i  \sigma_2(s)v_{j+1} , \notag
 \end{eqnarray}
\begin{eqnarray*}
M^{j,i}(u,T)&:=& D^{j,i}(u,T) \; \exp\Big(  \int_0^{\bar t_{i}} \Big( \theta^1_s( E_{j,i}(s,u,T) )+ \theta^2_s ( F_{j}(s,u)) \Big) ds  \Big)\notag\\
&& \times \prod_{l=2}^{j+1} \sqrt{\frac{\pi}{\beta \Delta t_{l}}}e^{-u_{l-1}^2/(4 \beta \Delta t_{l})}, \notag\\
N^{j,i}(v,T)&:=& \tilde D^{j,i}(v-\i R,T) \; \exp\Big(  \int_0^{\bar t_{i}} \Big( \theta^1_s( \tilde E_{j,i}(s,v- \i R,T) )+ \theta^2_s ( \tilde F_{j}(s,v- \i R)) \Big) ds  \Big)\notag\\
&& \times\frac{\exp\big(p(\bar t_i)(\i v_{j+1} +r)\big)}{(\i v_{j+1} +r-1)(\i v_{j+1}+r)}\prod_{l=2}^{j+1} \sqrt{\frac{\pi}{\beta \Delta t_{l}}}e^{-v_{l-1}^2/(4 \beta \Delta t_{l})}, \text{ for } v \in \mathbb{R}^{j+1}. \notag
\end{eqnarray*}

Finally, we define  for $0\leq s\leq T$, $u\in \mathbb{R}$ and $i\in\{1,\ldots,N\}$
\begin{eqnarray*}
E_1(s,u)&:=& \Sigma(s, \bar t_i)+ (r+\i u)(\beta(s)-\Sigma(s, \bar t_i)),\\
F_1(s,u) &:=& (r+ \i u) \sigma_2(s),\\
N^i(u)&:=& \exp\Big((r +\i u)w_{\bar t_i} \Big)\;\exp\Big( \int_0^{\bar t_i}\theta^1_s(E_1(s,u))ds + \int_0^{\bar t_i}\theta^2_s(F_1(s,u))ds\Big)\\
&& \times \frac{\exp\big(p(\bar t_i)(\i u +r)\big)}{(\i u +r-1)(\i u+r)}.
\end{eqnarray*}
 The following result for the DB value holds. 
\begin{theorem}\label{PDB}
The price ${\rm P}^{\rm DB}$ is given by
\begin{eqnarray*} 
{\rm P}^{\rm DB} &=& \sum_{i:\: \bar t_i \le t_1}Q( \tau^m(x) \in [\bar{t}_{i-1}, \bar{t}_i))( G(\bar t_i)B(0, \bar t_i) + G(\bar t_i)B(0, \bar t_i) A_{0,i})\\
&& + \sum_{j=1}^{K-2}\sum_{i:\:\bar t_i\in(t_j,t_{j+1}]}  Q\big( \tau^m(x) \in [\bar{t}_{i-1}, \bar{t}_i)\big)G(\bar t_i)B(0,\bar t_i)\big(A^1_{j,i} + A^2_{j,i}\big)\\
&& + \sum_{i:\: \bar t_i\in (t_{K-1},T]} Q\big( \tau^m(x) \in [\bar{t}_{i-1}, \bar{t}_i)\big)G(\bar t_i)B(0,\bar t_i)\big(A^1_{K-1,i} + A^2_{K-1,i}\big),
\end{eqnarray*}
where, using the notation from \eqref{DB notation}, for $j\in \{1,\ldots, K-1\}$
\begin{eqnarray*}
A_{0,i}&=&\frac{e^{-\int_0^{\bar t_i} A(s,\bar t_i)ds}}{2\pi}\int_{\mathbb{R}} N^i(u)du,\\
A_{j,i}^1&=& \frac{e^{-C(t_{j+1} -t_1)}}{(2\pi)^{j}}e^{-\int_0^{\bar t_i} A(s,\bar t_i)ds }\int_{\mathbb{R}^{j}}M^{j,i}(u,T) du,\\
A_{j,i}^2&=& \frac{e^{-C(t_{j+1} -t_1)}}{(2\pi)^{j+1}}e^{-\int_0^{\bar t_i} A(s,\bar t_i)ds }\int_{\mathbb{R}^{j+1}}N^{j,i}(u,T) du.
\end{eqnarray*}
\end{theorem}

The proof of Theorem \ref{PDB} is similar to the proof of Theorem \ref{prop:A1}. For this reasons we defer it to the  Appendix \ref{app:proofDB}.

\begin{remark} We have
\[
    E_{Q}\bigg[e^{-\int_0^{\bar t_i} r(u)du} {\rm DB}(\bar t_i)\bigg]
     =  Q( \tau^m(x) \in [\bar{t}_{i-1}, \bar{t}_i))E_{Q}\bigg[e^{-\int_0^{\bar{t}_i} r(u)du} \ind{ \tau^s \geq \bar{t}_i}\max(I S(\bar{t}_i), G(\bar{t}_i))\bigg]. 
  \]
As shown in the previous theorem this last expectation can be computed directly. However, this expression can also be traced back to ${\rm P}^{\rm GMAB}$. Note, that by definition,
  $$ E_{Q}\bigg[e^{-\int_0^{\bar{t}_i} r(u)du} \ind{ \tau^s  \geq \bar{t}_i}\max(I S(\bar{t}_i), G(\bar{t}_i))\bigg] 
     = \frac{{\rm P}^{\rm GMAB}(\bar{t}_i)}{Q( \tau^m(x)> \bar{t}_{i} )}. $$

To facilitate the computation of ${\rm P}^{\rm GMAB}(\bar{t}_i)$ we suggest the following approximation: first, for $j$ such that $t_j < \bar t_i \le t_{j+1}$, we obtain that
\begin{align*}
   {\rm P}^{\rm GMAB}(\bar{t}_i) = 
   Q( \tau^m(x)> \bar{t}_{i} )  E_{Q}\bigg[e^{-\int_0^{\bar{t}_i} r(u)du} e^{-\int_0^{t_{j+1}} \lambda^s(u) du}\max(I S(\bar{t}_i), G(\bar{t}_i))\bigg].
\end{align*}
Now, we define $\Hat{D}(\bar{t}_i)=Y(\bar{t}_i) - \delta \bar{t}_i$, then 
\[
\max(I S(\bar{t}_i), G(\bar{t}_i))= G(\bar{t}_i)\Big[1+ \Big(\exp(\hat{D}(\bar{t}_i)) -1\Big)^+\Big].
\]
The suggested approximation is
\begin{align*}
   {\rm \hat P}^{\rm GMAB}(\bar{t}_i) = 
   Q( \tau^m(x)> \bar{t}_{i} )G(\bar{t}_i)  E_{Q}\Big[e^{-\int_0^{\bar{t}_i} r(u)du} e^{-\int_0^{ t_{j+1}}  \lambda^s(u) du}\big(1+ \big(\exp(\hat{D}(\bar{t}_i)) -1\big)^+\big)\Big],
\end{align*}
which can be evaluated  using the results from Theorem \ref{prop:A1},  replacing $T$ by $\bar t_i$ and $K$ by the $j+1$ for which $t_j < \bar t_i \le t_{j+1}$.

\end{remark}

\subsection{Surrender benefit}
Finally, we compute the value of the surrender benefit 
$$ {\rm P}^{\rm SB} = \sum_{i=1}^{K-1}  E_{Q}\Big[   e^{-\int_0^{t_{i}} r(u)du} {\rm SB}(t_{i})\Big]. $$
 We use the spot measure, i.e.~the measure with the stock price  chosen as numeraire. This allows to exploit the dependence between the surrender intensity and the stock price. 

Recall the definition of $w_l$ from Equations \eqref{def:wk} and \eqref{def:wk2}. Further, for $0 \le s \le T$, $i\in \{2,\ldots, K-1\}$, $u \in \R^{i-1}$, and $v \in \R^{i}$, we define

 \begin{eqnarray*}
    D^{i}(u,T) &:=& \exp\bigg( \i\sum_{l=1}^{i-1} u_{l}w_l - \omega(t_{i}) \bigg), \notag\\
    \tilde D^{i}(v,T) &:=& D^{i}(v_1,\ldots,v_{i-1},T)  \exp\Big( \i\  v_{i}w_i  \Big), 
    \end{eqnarray*}
    \begin{eqnarray*}
  	E^{i}(s,u,T)&:=& \i(\beta(s) - \Sigma(s,T))\sum_{l=1}^{i-1} u_{l} \ind{0\leq s \leq t_{l}} 
	  + \beta(s),\notag \\
	\tilde E^{i}(s,v,T) &:=& E^{i}(s,v_1,\ldots,v_{i-1},T) + \i  (\beta(s) - \Sigma(s,T))v_{i},   
	\end{eqnarray*}
	\begin{eqnarray}\label{Ei}
  	F^{i}(s,u)&:=& \i\sigma_2(s)\sum_{l=1}^{i-1} u_{l} \ind{0\leq s \leq t_{l}}  + \sigma_2(s), \\
  	\tilde F^{i}(s,v) &:=& F^{i}(s,v_1,\ldots,v_{i-1}) + \i \sigma_2(s)v_{i},  \notag 
  	\end{eqnarray}
  	\begin{eqnarray*}
  	M^{i}(u,T) &:=& D^{i}(u,T) \;	\exp\Big(  \int_0^{t_{i}} \Big( \theta^1_s( E^{i}(s,u,T) )+ \theta^2_s ( F^{i}(s,u)) \Big) ds  \Big) \notag \\
  	&&  \times \prod_{l=2}^{i} \sqrt{\frac{\pi}{\beta \Delta t_{l}}}e^{-u_{l-1}^2/(4 \beta \Delta t_{l})}, \notag \\
    N^{i}(v,T) &:=& \tilde D^{i}(v,T) \;	\exp\Big(  \int_0^{t_{i}} \Big( \theta^1_s( \tilde E^{i}(s,v,T) )+ \theta^2_s ( \tilde F^{i}(s,v)) \Big) ds  \Big) \notag \\
  	&&  \times  \prod_{l=2}^{i+1} \sqrt{\frac{\pi}{\beta \Delta t_{l}}}e^{-v_{l-1}^2/(4 \beta \Delta t_{l})}.\notag
  	\end{eqnarray*}

\begin{theorem}\label{surrender} 
   The price ${\rm P}^{\rm SB}$ is given by  
   $$ {\rm P}^{\rm SB} = I \sum_{i=1}^{K-1} P(t_{i})Q(\tau^m(x) > t_{i}) (B_{i}^{1} - B_{i}^{2}), $$
   where $B_1^{1}=1$ and for $i\in \{2, \ldots, K-1\}$,
   \begin{eqnarray*}
      B_{i}^{1}&=& \frac{ e^{-C(t_{i} -t_1)} }{ (2\pi)^{i-1}} 
      \int_{\mathbb{R}^{i-1}} M^{i}(u,T) du, 
 \end{eqnarray*}
  and, for $i\in\{1, \ldots, K-1\},$ 
 \begin{eqnarray*}
      B_{i}^{2}&=&  \frac{ e^{-C(t_{i+1} -t_1)} }{ (2\pi)^{i}} 
      \int_{\mathbb{R}^{i}} N^{i}(u,T) du.
 \end{eqnarray*}
\end{theorem}

\begin{proof}
  By construction $ {\rm P}^{\rm SB} = \sum_{i=1}^{K-1}  E_{Q}\Big[   e^{-\int_0^{t_{i}} r(u)du} {\rm SB}(t_{i})\Big]$. As 
  $$  {\rm SB}(t_{i})=  \ind{\tau^s =  t_{i}} \ind{\tau^s < \tau^m(x)} IS_{t_{i}}P(t_{i}), $$ we are interested in computing the expression\small
  \begin{eqnarray}
      E_{Q}\bigg[  e^{-\int_0^{t_{i}}r_s ds}   \ind{  \tau^s =  t_{i}} \ind{\tau^s < \tau^m(x)}  S_{t_{i}}) \bigg] 
   & =&E_{Q}\bigg[  e^{-\int_0^{t_{i}}r_s ds}   \ind{   \tau^s=t_{i} } \ind{ t_{i} < \tau^m(x)}  S_{t_{i}} \bigg] \notag \\
    & =& \: Q(\tau^m(x) > t_{i}) \; E_{Q}\bigg[  e^{-\int_0^{t_{i}}r_s ds}   \ind{   \tau^s =t_{i} }   S_{t_{i}} \bigg]. \label{temp35}
   \end{eqnarray}\normalsize
   It follows from \eqref{def:taus} that, for $A\in \mathscr{F}_{t_{i}}^{L^1,L^2}\subset \mathscr{F}_{t_{i+1}}^{L^1,L^2}$,
   \begin{eqnarray*}
   	  E_Q \Big[ \Ind_A \ind{\tau^s = t_{i} } \Big]  
   	  &=&
   	  E_Q \Big[ \Ind_A (\ind{t_{i} \le \tau^s} -  \ind{ t_{i+1}\leq \tau^s}) \Big] \\
   	  &=& E_Q \Big[ \Ind_A \Big( e^{-\int_0^{t_{i}} \lambda^s(u) du } -  e^{-\int_0^{t_{i+1}} \lambda^s(u) du }\Big) \Big].
   \end{eqnarray*}
   Consequently, equation \eqref{temp35} can be written as
   \begin{eqnarray*}
     E_{Q}\bigg[  e^{-\int_0^{t_{i}}r_s ds}   \ind{  \tau^s =  t_{i}} \ind{\tau^s < \tau^m(x)}  S_{t_{i}}) \bigg] 
   	  &=& Q(\tau^m(x) > t_{i}) \; E_{Q}\bigg[  e^{-\int_0^{t_{i}}r_u du} S_{t_{i}} e^{-\int_0^{t_{i}} \lambda^s(u) du }      \bigg] \\
     &&- \: Q(\tau^m(x) > t_{i}) \; E_{Q}\bigg[  e^{-\int_0^{t_{i}}r_u du} S_{t_{i}} e^{-\int_0^{t_{i+1}} \lambda^s(u) du }      \bigg].
  \end{eqnarray*}
$Q(\tau^m(x) > t_{i})$ is as in equation \eqref{Q^m}.	Let us write $B^1_{i}$ for the first expectation and $B^2_{i}$ for the second one. In order to compute these two expectations we introduce the spot probability measure $Q^{S,i},\ i=1,\dots,K-1$ defined by its Radon-Nikodym derivative
   \begin{equation}
  \frac{dQ^{S,i}}{dQ} = e^{-\int_0^{t_{i}}r_u du} S(t_{i}). \label{bis forward measure}
  \end{equation}
 The above defines indeed a density process as the discounted stock price $\big(e^{-\int_0^{t}r_u du}S(t)\big)_t$ is a $Q$-martingale. We can use the new measure to simplify  $B^1_{i}$ and $B^2_{i}$ as follows
  \begin{eqnarray}
   B^1_i &=&E_{Q}\bigg[  e^{-\int_0^{t_{i}}r_u du}  S_{t_{i}}  e^{-\int_0^{t_{i}} \lambda^s(u) du }    \bigg] \notag\\
     &=&  E_{Q^{S,i}}\bigg[ e^{-\int_0^{t_{i}} \lambda^s(u) du }    \bigg], \label{B^1} \\
   B^2_i &=&E_{Q}\bigg[  e^{-\int_0^{t_{i}}r_u du}  S_{t_{i}} e^{-\int_0^{t_{i+1}} \lambda^s(u) du }     \bigg] \notag\\
     &=&  E_{Q^{S,i}}\bigg[ e^{-\int_0^{t_{i+1}} \lambda^s(u) du }    \bigg]. \label{B^2}
  \end{eqnarray}
  Note that $\lambda^s(u)$ for $t_i\leq u < t_{i+1}$ is defined by $D(t_i)$ and consequently it is $\mathscr{F}_{t_i}^{L^1,L^2}$-measurable.
Focusing first on $B^1_i$, by construction  $\lambda^s(u)=0$ for $u\in[0, t_1)$, and therefore  $B_1^{1}=1$.
For $i\in \{2, \ldots, K-1\}$, we follow the same strategy as in the proof of Theorem \ref{prop:A1} for the computation of $A_1$ and consider
  \begin{align}   		e^{C (t_{i} - t_1) } E_{Q^{S,i}} \bigg[ e^{-\int_0^{t_{i}} \lambda^s(u) du }    \bigg] 
 		&= E_{Q^{S,i}}\bigg[f(D(t_1), \ldots, D(t_{i-1}))\bigg], \label{B1}
 \end{align}
 where $f(x_1, \ldots, x_{i-1}):=\prod_{l=2}^{i} e^{-\beta \Delta t_{l} x_{l-1}^2}$ with $\Delta t_l= t_l  - t_{l-1}$.
 As in \eqref{Price-Fourier} we obtain that the last expectation is
	\begin{equation}
		\frac{1}{(2\pi)^{i-1}}\int_{\mathbb{R}^{i-1}}\tilde M^{i-1}(\i u)\hat{f}(-u)du, \label{SE}
	\end{equation}
  with
	\begin{equation*}
 		\hat{f}(u_1, \ldots, u_{i-1})=  \prod_{l=2}^{i} \sqrt{\frac{\pi}{\beta \Delta t_{l}}}e^{-u_{l-1}^2/(4 \beta \Delta t_{l})},
	\end{equation*}  
  and $\tilde M^{i-1}(\i u)$ defined as 
    \begin{align*}
  		 \tilde M^{i-1}(\i u) &= E_{Q^{S,i}}\bigg[e^{\i u_1 D(t_1) + \ldots + \i u_{i-1}D(t_{i-1})}\bigg].
  	\end{align*}
  From \eqref{D:appendix} it follows that
  	\begin{eqnarray*}
		\tilde M^{i-1}(\i u) &=&
		E_{Q}\bigg[e^{\i u_1 D(t_1) + \ldots + \i u_{i-1}D(t_{i-1})+ \int_0^{t_{i}} \sigma_2(s)dL^2_s + \int_0^{t_{i}}\beta(s)dL^1_s -\omega(t_{i})}\bigg]\\
		&=& \exp\Big( \i\sum_{l=1}^{i-1} u_{l}(-p(t_{l}) - \delta T+ \int_0^T f(0,s)ds + \int_0^{t_{l}}A(s,T)ds -\omega(t_{l})) \Big)\\
	 	&& \times E_{Q}\bigg[
	 	   \exp\bigg( \i\sum_{l=1}^{i-1} \Big(\int_0^{t_{l}} u_{l}\sigma_2(s)dL^2_s + \int_0^{t_{l}}u_{l}(\beta(s) - \Sigma(s,T))dL^1_s\Big) \\
	 	&&   \qquad \qquad \qquad +  \int_0^{t_{i}} \sigma_2(s)dL^2_s + \int_0^{t_{i}}\beta(s)dL^1_s -\omega(t_{i}) \bigg)   \bigg].
	 \end{eqnarray*}
  By equation \eqref{cumulant} it follows that
  \begin{align*}
  	  \lefteqn{ E_{Q}\Big[\exp\Big(  \int_0^{t_{i}} E^{i}(s,u,T) dL^1_s + \int_0^{t_{i}} F^{i}(s,u) dL^2_s \Big) \Big] } \qquad \qquad \\ 
  	  &= \exp\Big(  \int_0^{t_{i}} \Big( \theta^1_s( E^{i}(s,u,T) )+ \theta^2_s ( F^{i}(s,u)) \Big) ds  \Big) 
  \end{align*}
  Therefore, given the definition of $D^{i}(u,T)$ in \eqref{Ei} we have
  \begin{align*}
  	\tilde M^{i-1}(\i u) &= D^{i}(u,T) 	\exp\Big(  \int_0^{t_{i}} \Big( \theta^1_s( E^{i}(s,u,T) )+ \theta^2_s ( F^{i}(s,u)) \Big) ds  \Big) 
  \end{align*}
  and the representation of $B^1_{i}$ follows from \eqref{SE}.

 Finally, we note that
	 \begin{eqnarray}
 		B_{i}^{2} &=& e^{-C(t_{i+1} -t_1)} E_{Q^{S,i}} \bigg[\prod_{l=2}^{i+1}  e^{-\beta \Delta t_l D(t_{l-1})^2} \bigg].
	 \end{eqnarray}
	 Repeating mutatis mutandis the above arguments, the expression for $B^2_i$ follows.
\end{proof}

\section{Numerical implementation and testing}\label{sec:numerics}

This section discusses the numerical computation of the pricing equations in Theorems \ref{prop:A1}, \ref{PDB} and \ref{surrender}, based on the model features introduced in section \ref{sec:market}. 
For the purpose of the numerical analysis, we choose as a relevant L\'evy process the Normal Inverse Gaussian (NIG) process introduced by \cite{BN97} with cumulant function 
\begin{equation*}
\theta(u)=\mu u+\delta\left(\sqrt{\alpha^2-\beta^2}-\sqrt{\alpha^2-(\beta+u)^2}\right), \quad -\alpha-\beta<u<\alpha-\beta,
\label{CF:NIG}
\end{equation*}
for $\mu\in\mathbb{R}$, $\delta>0$, $0\leq |\beta|<\alpha$. The parameter $\alpha$ controls the steepness of the density (and therefore its tail behaviour), $\beta$ primarily controls the (sign of the) skewness of the distribution, whilst  $\delta$ is the scale parameter; the location parameter $\mu$ is instead set to zero, without loss of generality.

Further, we assume a simplified Vasi\v{c}ek structure for the function $\sigma_1(s,T)$ so that for $a>0$
\[\sigma_1(s,T)=\left\{\begin{array}{ll}
a e^{-a(T-s)},& s\leq T\\
0, & s>T
\end{array}\right.\]
and
\[\Sigma(s,T)=\left\{\begin{array}{ll}
1-e^{-a(T-s)},& s\leq T\\
0, & s>T
\end{array}\right. .\]
For the equity part, we assume $\sigma_2(s)=\sigma_2>0$ and $\beta(s)=b\in \mathbb{R}$.
\begin{table}[tbp]
  \centering
  \caption{\small Parameters for the Variable Annuity contract. Reference process: NIG. Financial Market Model parameters - source: \cite{ER}. Mortality Model parameters - source: \cite{Escobarea2016}.}
  \resizebox{0.99\textwidth }{!}{
  \begin{tabular}{rrrlcrrrrrlr}\toprule
    \multicolumn{2}{c}{Variable Annuity} &       & \multicolumn{3}{c}{Financial Market Model} &       & \multicolumn{2}{c}{Surrender Model} &       & \multicolumn{2}{c}{Mortality Model} \\\midrule
    &       &       &       & \multicolumn{1}{c}{$L_t^1$} & \multicolumn{1}{c}{$L_t^2$} &       &       &       &       &       &  \\
    \multicolumn{1}{l}{$T$} & \multicolumn{1}{l}{3, 4, 10 years} &       & $\alpha$ & \multicolumn{1}{r}{4} & 5.73  &       & \multicolumn{1}{l}{$\beta$} & 0.05  &       & $b$     & 12.1104 \\
    \multicolumn{1}{l}{$\delta$} & \multicolumn{1}{l}{0.01 p.a.} &       & $\beta$ & \multicolumn{1}{r}{-3.8} & -2.13 &       & \multicolumn{1}{l}{$C$} & 0.01  &       & $z$     & 76.139 \\
    \multicolumn{1}{l}{$P(t_l)$} & \multicolumn{1}{l}{$0.95+0.05 t_l/T$} &       & $\delta$ & \multicolumn{1}{r}{1.34} & 8.3   &       &       &       &       & $\kappa$ & 0.4806 \\
    \multicolumn{1}{l}{$\Delta t_l$} & \multicolumn{1}{l}{1 year} &       & $a$     & \multicolumn{1}{r}{0.0020898} & \multicolumn{1}{c}{-} &       &       &       &       & $\lambda$ & 0.0195 \\
  $\bar{t}_i-\bar{t}_{i-1}$ & \multicolumn{1}{l}{6 months}      &       & $\sigma_2$     & -     & 0.1818 &       &       &       &       & $\sigma$ & 0.0254 \\
   &       &       & $b$     & -     & 0.0065 &       &       &       &       &  &  \\
    \bottomrule
  \end{tabular}}
  \label{tab:NIGParams}
\end{table}%
All numerical experiments refer to contracts with a parameter setting as in Table \ref{tab:NIGParams}; the parameters of the financial model are taken from the calibration exercise of \cite{ER}. The short maturity contracts are used for benchmarking purposes, whilst the 10 year maturity contract represents a realistic specification for practical purposes. Concerning the possible termination dates, for illustration purposes we use a 1 year frequency, i.e. $\Delta t_l=1$, whilst the time grid for mortality, which is assumed to have finer granularity, is built with a frequency of 6 months. We note though that the results obtained in Theorems \ref{prop:A1}, \ref{PDB} and \ref{surrender} hold for any choice of the relevant time steps. The numerical schemes are implemented in  \textsc{Matlab R2018a} and run on a computer with an Intel i5-6500, 3.20 gigahertz CPU, and 8 gigabytes of RAM.

\subsection{Implementation}
The multidimensional integrals in Theorems \ref{prop:A1}, \ref{PDB} and \ref{surrender} are computed by means of Monte Carlo integration \citep[see for example][and references therein]{PHARR2010636}. Thus, the Monte Carlo estimate of a high-dimensional integral over the domain $\Omega$
\[I=\int_{\Omega}f(x)dx\]
is obtained by evaluating the function $f(x)$ at $M$ points $x$, drawn randomly in $\Omega$ with a given probability density $p(x)$, so that
\[I_{MC}=\frac{1}{M}\sum_{i=1}^{M}\frac{f(x_{i})}{p(x_{i})}.\]
The error is measured by means of the (unbiased) sample variance
\[\frac{\sum_{i=1}^M(I_{i}-I_{MC})^2}{M-1},\]
as in standard Monte Carlo simulation; this point shows that the rate of convergence is independent of the dimensionality of the integrand, which is the advantage of Monte Carlo integration compared to standard numerical quadrature techniques in computing high-dimensional integrals. 
\begin{figure}[tbp]
  \centering
  \includegraphics[width=0.49\textwidth,height=0.25\textheight]{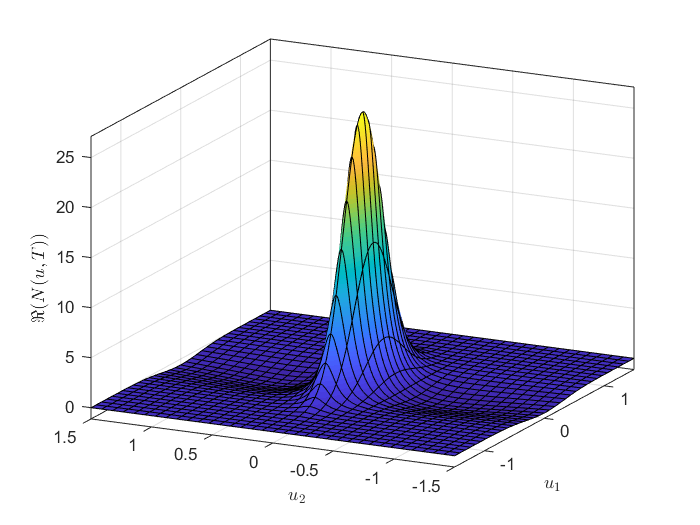}
  \includegraphics[width=0.49\textwidth,height=0.25\textheight]{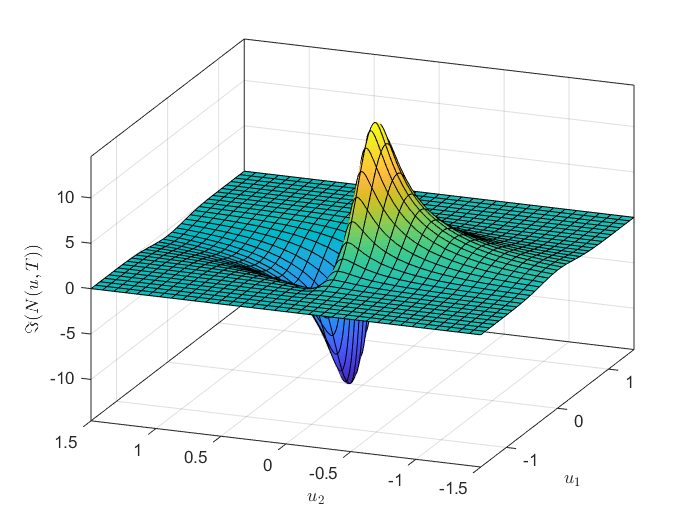}
  \caption{\small $N(u,T)$, $T=3$ years ($K=2$). Left panel: real part of the function $N(u,T)$. Right panel: imaginary part of $N(u,T)$.  Parameters: Table \ref{tab:NIGParams}.  
     \label{fig:MuTNuTK2}}
\end{figure}

Classical Monte Carlo integration uses the uniform probability density. However, visual inspection of the integrand functions $M(u,T)$ and $N(u,T)$ of Theorem \ref{prop:A1} reveals that $M(u,T)$ is strongly peaked around the origin, whilst $N(u,T)$ is characterised by alternations of peaks and troughs of significant magnitude -  as shown in Figure \ref{fig:MuTNuTK2} for $T=3$ ($K=2$). 
Similar considerations hold for the integrand functions in Theorems \ref{PDB} and \ref{surrender}, due to the similarities between the corresponding payoff functions. These features, in general, can induce large variances in the Monte Carlo estimate of the corresponding integral.

Consequently, for variance reduction purposes and in order to speed up convergence, we implement importance sampling; in details, given that the integrand functions are strongly peaked around the origin, the importance distribution of choice is the multivariate Gaussian distribution with zero mean (i.e. centered around the peak), independent components, and a given variance matrix, which is treated as a parameter. Full deployment of the pricing algorithm in the setting of VEGAS and MISER Monte Carlo integration is left to future research.

\subsection{Benchmarking and Testing}
Numerical results for the pricing functions considered in this paper are reported in Table \ref{tab:GMABbench} and \ref{tab:fullVA}. In order to provide a reliable benchmark for the Monte Carlo integration procedure introduced above, we consider first some simple examples for which the relevant multidimensional integrals can be tackled with standard quadrature packages. We illustrate the results obtained for the case of the GMAB. Similar performances are obtained for the DB and the SB as well, and we refer the interested reader to the  Appendix \ref{App:bench} for more details.
\begin{table}[tbp]
  \centering
  \caption{\small Benchmarking Monte Carlo integration with importance sampling - the case of the GMAB. Parameters: Table \ref{tab:NIGParams}. `Quadrature':  Matlab built-in functions {\fontfamily{qcr}\selectfont integral}, {\fontfamily{qcr}\selectfont integral2} and {\fontfamily{qcr}\selectfont integral3}. Bias/standard error expressed as percentage of the actual value. Monte Carlo iterations: 100 batches of size $10^6$. CPU time expressed in seconds and referred to the average time of 1 batch of $10^6$ iterations.}
  \resizebox {0.80\textwidth }{!}{
    \begin{tabular}{lclccccr}\toprule
       &       &       & Quadrature & \multicolumn{4}{c}{Monte Carlo integration (Imp. Sampling)} \\
      \multicolumn{1}{c}{$T$} & $K$     &       & Value & Value & Bias (\%) & Std. Error (\%) &  \\\midrule
      \multicolumn{1}{l}{3 years} & 2     & $A_1$  & 0.9867 & 0.9867 & 0.0035 & 0.0050 &  \\
      &       & $A_2$  & 0.1487 & 0.1482 & 0.3584 & 0.2683 &  \\
      &       & CPU   & 6.6323 & 31.2944 & &       &  \\
      &       &       & \multicolumn{1}{l}{($A_1$: 0.1280)} &       &       &       &  \\
      &       &       & \multicolumn{1}{l}{($A_2$: 6.5043)} &       &       &       &  \\
      &       &       &       &       &       &       &  \\
      \multicolumn{1}{l}{4 years} & 3     & $A_1$  & 0.9703 & 0.9702 & 0.0139 & 0.0076 &  \\
      &       & $A_2$  & 0.1669 & 0.1669 & 0.0155 & 0.0647 &  \\
      &       & CPU   & 589.0926 & 51.6269 & &       &  \\
      &       &       & \multicolumn{1}{l}{($A_1$: 1.3042)} &       &       &       &  \\
      &       &       & \multicolumn{1}{l}{($A_2$: 587.7884)} &       &       &       &  \\
      \bottomrule
    \end{tabular}}
  \label{tab:GMABbench}
\end{table}
\begin{table}[tbp]
  \centering
  \caption{\small Variable Annuity by Monte Carlo integration with importance sampling. Parameters: Table \ref{tab:NIGParams}. Standard error expressed as percentage of the actual value. Monte Carlo iterations: 100 batches of size $10^6$. CPU time expressed in seconds and referred to the average time of 1 batch of $10^6$ iterations.}
  \begin{tabular}{lrrrr}\toprule
    \multicolumn{1}{c}{$T$} & \multicolumn{1}{c}{GMAB} & \multicolumn{1}{c}{DB} & \multicolumn{1}{c}{SB} & \multicolumn{1}{c}{VA} \\\midrule
    4 years & 112.5121 & 4.9280 & 2.6997 & 120.1399 \\
    (Std. Error \%) & 0.0001 & 0.0256 & 0.1142 &  \\
    CPU   & 51.6269 & 264.3202 & 31.7341 &  \\
    &       &       &       &  \\
    10 years & 93.0783 & 14.2344 & 15.4533 & 122.7661 \\
    (Std. Error \%) & 0.0003 & 0.0098 & 0.0534 &  \\
    CPU   & 179.0842 & 1971.3826 & 717.3945 &  \\\bottomrule
  \end{tabular}
  \label{tab:fullVA}%
\end{table}%

In details. Table \ref{tab:GMABbench} reports values obtained with deterministic quadrature methods, the corresponding estimate from Monte Carlo integration, and the CPU time. Together with the Monte Carlo estimate, $I_{MC}$, we also report measures of its accuracy in terms of the absolute value of the bias of the estimator expressed as percentage of the value obtained by quadrature $I_{Q}$, i.e.
\[100\times \frac{|I_{MC}-I_{Q}|}{I_{Q}}.\]
In addition, we also report the percentage standard error of the Monte Carlo estimate
\[100\times\frac{1}{I_{MC}}\sqrt{\frac{\sum_{i=1}^M(I_{i}-I_{MC})^2}{M(M-1)}}.\]
For increased reliability of the estimate of the CPU time, we base the numerical experiment on 100 repetitions of the Monte Carlo algorithm with sample size $10^6$.

We consider two examples for a contract with maturity $T=3$ years (i.e. $K=2$) and $T=4$ years (i.e. $K=3$) respectively. With annually spaced termination dates $t_i$, these choices imply that $A_1$ and $A_2$ are respectively one- and two-dimensional integrals in the first case, and two- and three-dimensional integrals in the second case, which can also be computed using packages for deterministic quadrature methods. 

For importance sampling, numerical experiments show that relatively small biases and standard errors can be obtained by using the same variance fixed at 0.25 for the first $K-1$ dimensions, and increasing this value to 1 for the last $K^{th}$ dimension, as to cater for the higher variability of the integrand function $N(u,T)$. The results in Table \ref{tab:GMABbench} confirm the quality of the estimates as all biases and standard errors are below 0.5\%.

The value of the VA together with its components are reported in Table \ref{tab:fullVA}. We observe that the GMAB value decreases with increasing maturities, due to the larger number of dates at which the contract can be terminated by either death or surrender of the policyholder. Consistently with these findings, we also observe the increase in the value of both the DB and the SB for contracts with longer maturities. The overall value of the VA increases as well with longer maturities, denoting the dominant impact of the DB and the SB components.

Finally, we observe the higher computational cost of the DB and SB parts of the full variable annuity; this is due to the large number of integrals of the cumulant functions of interest (see definitions (\ref{DB notation}) and (\ref{Ei})) be computed. This number depends on the granularity of both the $\mathbf{t}$ and $\bar{\mathbf{t}}$ grids: the finer these grids, the more computational demanding these components of the contract become.

\section{Results: Sensitivity Analysis}\label{sec:sensitivity}

The sensitivity analysis is carried out by perturbing the parameters of interest one at a time, ceteris paribus. The benchmark case is given by a 10 years contract with the parameters set as in Table \ref{tab:NIGParams}.

\begin{figure}[tbp]
  \centering
  \includegraphics[width=0.49\textwidth,height=0.20\textheight]{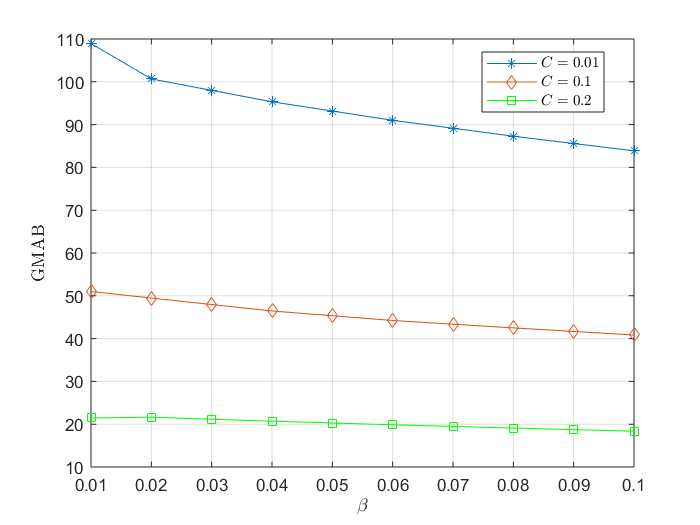}
  \includegraphics[width=0.49\textwidth,height=0.20\textheight]{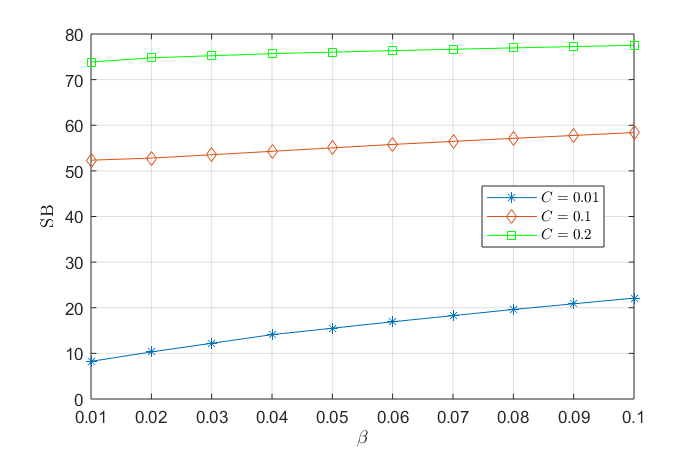}
  \includegraphics[width=0.49\textwidth,height=0.20\textheight]{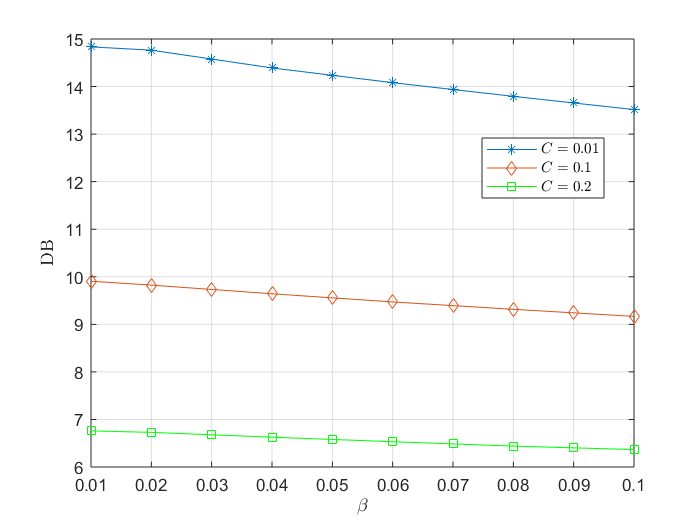}
  \includegraphics[width=0.49\textwidth,height=0.20\textheight]{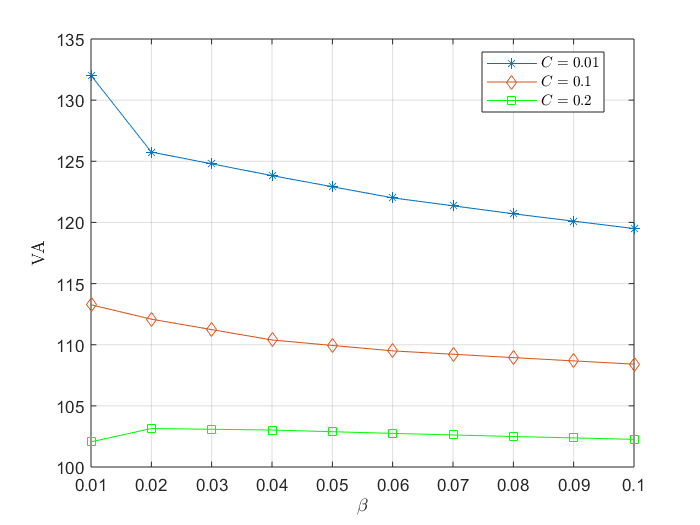}
  \caption{\small Sensitivity Analysis: the surrender parameters $(\beta,C)$. Top panels: left-hand-side - GMAB; right-hand-side: SB. Bottom panels: left-hand-side - DB; right-hand-side - VA. Maturity: $T=10$ years. Other parameters: Table \ref{tab:NIGParams}. \label{fig:betaC}}
\end{figure}
\begin{figure}[htbp]
  \centering
  \includegraphics[width=0.49\textwidth,height=0.20\textheight]{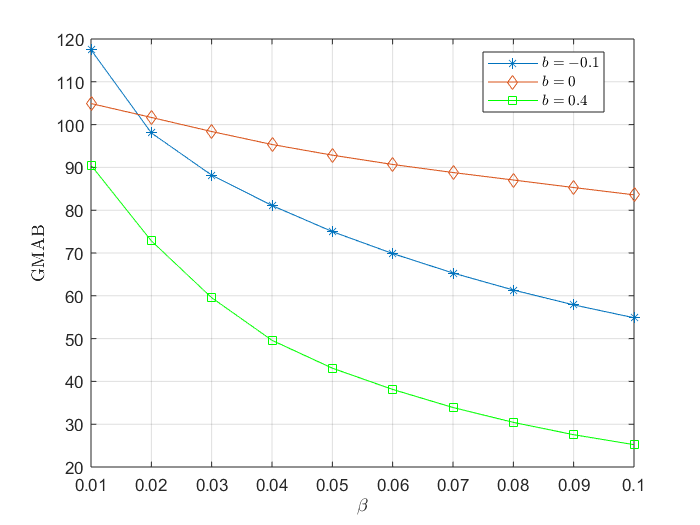}
  \includegraphics[width=0.49\textwidth,height=0.20\textheight]{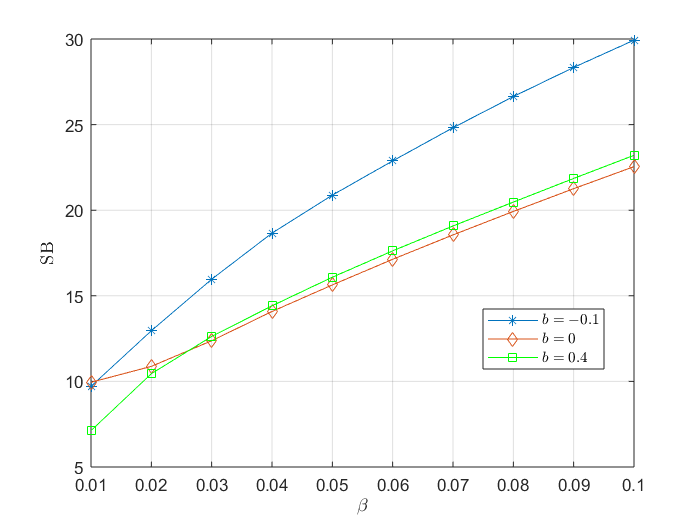}
  \includegraphics[width=0.49\textwidth,height=0.20\textheight]{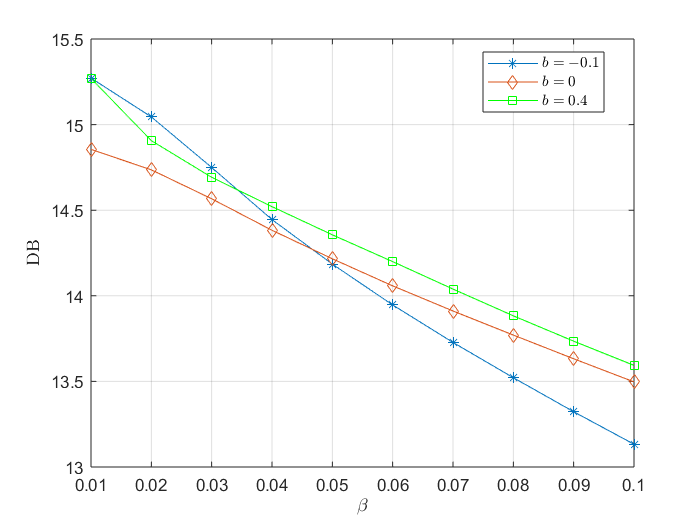}
  \includegraphics[width=0.49\textwidth,height=0.20\textheight]{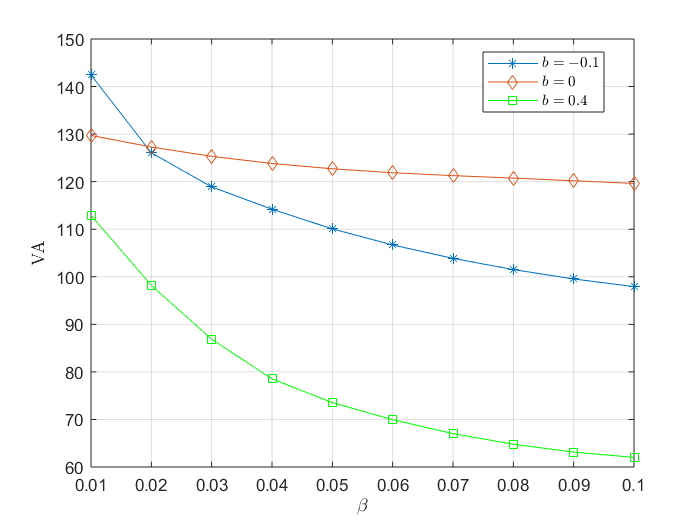}
  \caption{\small Sensitivity Analysis: the dependence parameters $b$ and $\beta$. Top panels: left-hand-side - GMAB; right-hand-side: SB. Bottom panels: left-hand-side - DB; right-hand-side - VA. Maturity: $T=10$ years. Other parameters: Table \ref{tab:NIGParams}. \label{fig:betab}}
\end{figure}
\begin{figure}[tbp]
  \centering
  \includegraphics[width=0.49\textwidth,height=0.20\textheight]{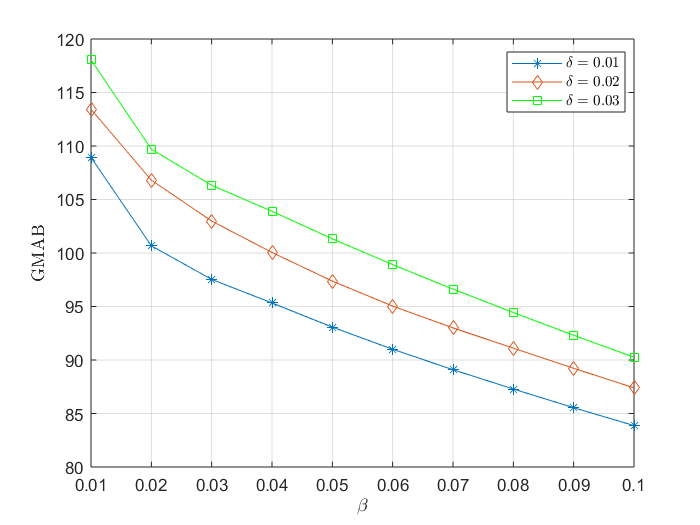}
  \includegraphics[width=0.49\textwidth,height=0.20\textheight]{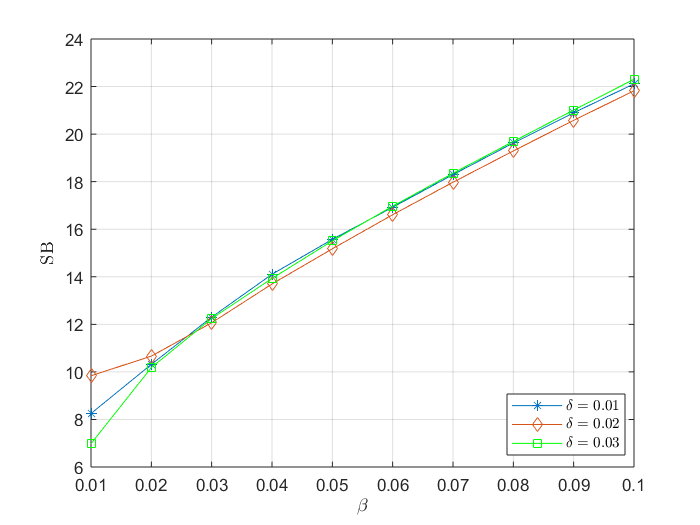}
  \includegraphics[width=0.49\textwidth,height=0.20\textheight]{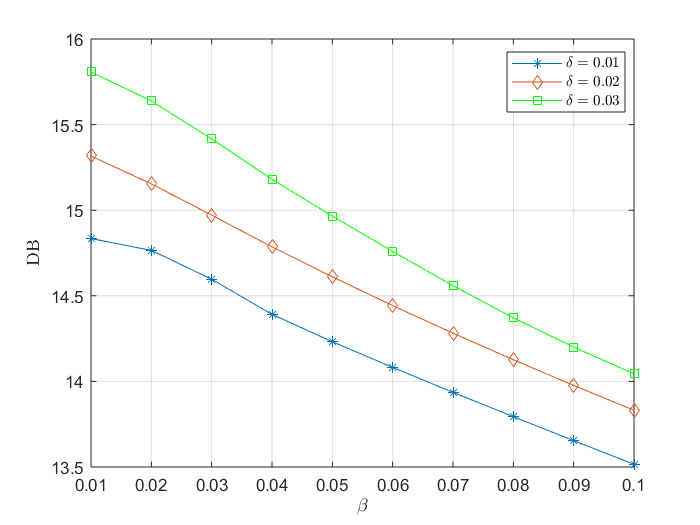}
  \includegraphics[width=0.49\textwidth,height=0.20\textheight]{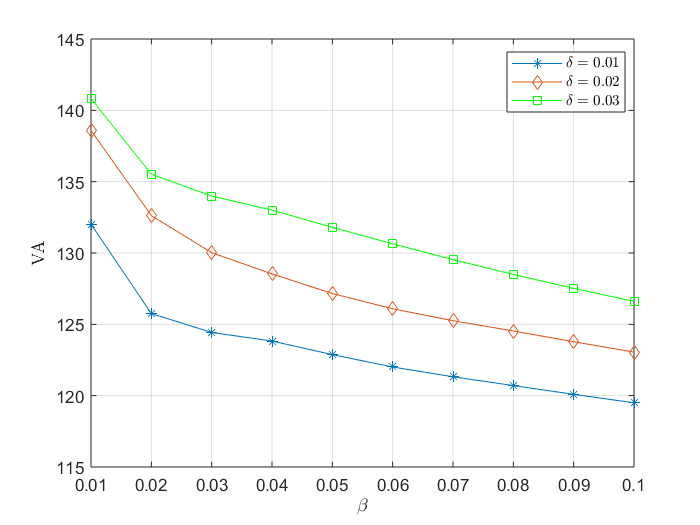}
  \caption{\small Sensitivity Analysis: the guarantee $\delta$ for different values of $\beta$. Top panels: left-hand-side - GMAB; right-hand-side: SB. Bottom panels: left-hand-side - DB; right-hand-side - VA. Maturity: $T=10$ years. Other parameters: Table \ref{tab:NIGParams}. \label{fig:betag}}
\end{figure}

In Figure \ref{fig:betaC} we show the sensitivity of the VA and its components to the parameters $\beta$ and $C$, controlling the surrender intensity. Consistently with intuition, the higher the value of $\beta$, the stronger the influence of the financial market on the policyholder decision to surrender. This is captured in particular by the value of the surrender benefit, SB, which increases with $\beta$. Due to the higher surrender probability, the values of the GMAB and the DB reduce, and so does the resulting value of the VA. 

The impact of the parameter $\beta$ though becomes less relevant in correspondence of higher values of the constant $C$ which represents the baseline surrender behaviour. In this case, the non-economic factors leading to the decision of surrender dominate the influence of the financial market to the point that the VA as well as its components are almost insensitive to $\beta$.

These results highlight that a correct quantification of the baseline surrender parameter $C$ is of paramount importance given the significant impact on the value of the VA and its components. In this respect, it is crucial for insurance companies to correctly calibrate the surrender model to the information that they could collect regarding the policyholder behaviour towards lapses.

Figure \ref{fig:betab} shows that the impact of the parameter $\beta$ is also affected by the dependence between the equity and the fixed income markets, here captured by the parameter $b$. Indeed, in presence of explicit dependence between the financial markets (i.e. $b\neq 0$), the behaviour of the policyholders becomes more responsive to changes in the market conditions. This is reflected in the rate of change of the contract's value.

Finally, in Figure \ref{fig:betag} we illustrate the behaviour of the contract with respect to the guaranteed rate $\delta$. A higher value of this rate corresponds to a more valuable VA, as one would expect; the SB though is relatively insensitive to $\delta$, also due to the low values of the latter used in this analysis.

\section{Conclusion}\label{Conclusion}
We proposed a general framework for the valuation of a number of guarantees commonly included in variable annuity products, such as guaranteed minimum accumulation benefit, death benefit and surrender benefit, in the setting of a hybrid model based on multivariate L\'evy processes, with surrender risk captured endogenously. 

This framework proves to be tractable, and allows for the deployment of efficient numerical schemes based on Monte Carlo integration. We used this setting to gain insights into the contract sensitivity with respect to the model parameters, with special focus on surrender risk. The emphasis on surrender in particular finds its motivation in the large percentage of early terminations experienced by the insurance companies issuing variable annuities, and the resulting net cash outflow they sustain as a consequence (as also pointed out by LIMRA Secure Retirement Institute and IRI). The results obtained in our analysis confirm  the importance of the surrender behaviour and its appropriate modelling.

We envisage the applicability of the model setup offered in this paper, and the proposed numerical scheme, also for other benefits offered by VAs, such as the post-retirement guaranteed minimum income benefit. This would represent a promising route for further research.

\section*{Acknowledgements}
The authors gratefully acknowledge the financial support from the Freiburg Institute for Advanced Studies (FRIAS) within the FRIAS Research Project `Linking Finance and Insurance: Theory and Applications (2017--2019)'.

 \begin{appendix}

\section{An alternative specification for the surrender intensity}\label{sec:altsurrender}
  The surrender intensity $\lambda^s$ is defined in \eqref{lambda}  by \begin{equation} 
   \lambda^s(t)=\beta D^2(t_i) +C=W_1(D(t_i)), \:\:\: t_i\leq t < t_{i+1} \label{W1}   
  \end{equation}
  where $W_1(x):=\beta x^2+C$. Consequently $W_1(x)\to \infty$ as $|x|\to \infty$. One might be interested in the behaviour of the contract price for a surrender intensity whose values do not become arbitrarily large. Therefore, in the following we study the valuation of the variable annuity's components under an alternative specification which keeps the surrender intensity bounded.
  
  We start with an observation concerning the behaviour of $D(t_i)$, $i=1,\dots,K-1$. The following lemma shows that these quantities stay in a bounded interval with high probability. Based on this fact, Proposition \ref{Prop:barA1} and its corollary show that the contract values obtained with the alternative surrender intensity approximate the values under the original assumption.

  \begin{lemma}\label{bislambda} For any $\epsilon >0$, there exists $L>0$ such that 
\begin{eqnarray}
&&Q^T(\bigcup_{i=1}^{K-1} \{|D(t_{i})|>L\}) \leq \epsilon, \notag\\
&& Q^{S,j}(\bigcup_{i=1}^{j} \{|D(t_{i})|>L\}) \leq \epsilon, \text{ for any } j\in\{1, \ldots, K-1\},\label{bounded}\\
&& Q^{\bar{t}_{i}}(\bigcup_{l=1}^{j} \{|D(t_{l})|>L\}) \leq \epsilon,\notag
\end{eqnarray}
 where the last result holds for  $j\in\{1,\ldots, K-2\}$ and any $i$  such that $t_j< \bar t_i \le t_{j+1}$, and for $j=K-1$ and any $i$ such that $t_{K-1}< \bar t_i \le T$. $Q^{\bar{t}_{i}}$ is defined in \eqref{second forward measure}, $Q^{S,j}$ is defined in \eqref{bis forward measure}, and $Q^T$ is defined in \eqref{forward measure}.
\end{lemma}

\begin{proof}

Let $L_1$ denote a positive constant. Applying the Markov inequality  and the inequality $e^{|x|}\leq e^x+e^{-x}$, we get
\begin{eqnarray}
Q^T(\bigcup_{i=1}^{K-1} \{|D(t_{i})|>L_1\})&\leq& \sum_{i=1}^{K-1} Q^T( \{|D(t_{i})|>L_1\})=\sum_{i=1}^{K-1} Q^T( \{e^{|D(t_{i})|}>e^{L_1}\})\notag\\
&\leq& e^{-L_1} \sum_{i=1}^{K-1} \big[ E^T(e^{D(t_{i})}) + E^T(e^{-D(t_{i})})\big].\label{negligible}
\end{eqnarray}
Now we use the representation given in \eqref{D:appendix}, the definition of $Q^T$ in \eqref{forward measure} and  \eqref{Bank ac}, to derive
\begin{eqnarray*}
E^T(e^{D(t_{i})}) &=& e^{-p(t_{i}) - \delta T+ \int_0^T f(0,s)ds + \int_0^{t_{i}}A(s,T)ds -\omega(t_{i})}E^T(e^{\int_0^{t_{i}} \sigma_2(s)dL^2_s+\int_0^{t_{i}}(\beta(s)-\Sigma(s,T)dL^1_s})\\
&=& e^{-p(t_{i}) - \delta T+ \int_0^T f(0,s)ds + \int_0^{t_{i}}A(s,T)ds -\omega(t_{i})-\int_0^T A(s,T)ds}\\
&& \times E_Q(e^{\int_0^{t_{i}} \sigma_2(s)dL^2_s+\int_0^{t_{i}}(\beta(s)-\Sigma(s,T)dL^1_s+\int_0^T\Sigma(s,T)dL^1_s})
\end{eqnarray*}
Observe that, by independence of $L^1$ and $L^2$, we have
\begin{eqnarray*}
&&E_Q(e^{\int_0^{t_{i}} \sigma_2(s)dL^2_s+\int_0^{t_{i}}(\beta(s)-\Sigma(s,T)dL^1_s+\int_0^T\Sigma(s,T)dL^1_s})\\
&&=E_Q(e^{\int_0^{t_{i}} \sigma_2(s)dL^2_s+\int_0^T((\beta(s)-\Sigma(s,T))\ind{0\leq s \leq t_{i}}+ \Sigma(s,T))dL^1_s})\\
&& = E_Q(e^{\int_0^{t_{i}} \sigma_2(s)dL^2_s}) E_Q(e^{\int_0^T((\beta(s)-\Sigma(s,T))\ind{0\leq s \leq t_{i}}+ \Sigma(s,T))dL^1_s}),
\end{eqnarray*}
where the last quantities are finite due to \eqref{AssM} and \eqref{assM12}. Therefore, $C_i:=E^T(e^{D(t_{i})})$ is finite, and by similar arguments $C_i':=E^T(e^{-D(t_{i})})$ is finite as well. Observe that $e^{-L_1}\sum_{i=1}^{K-1}(C_i+C_i') \leq \epsilon$ is equivalent to $L_1\geq -\log(\epsilon)+ \log(\sum_{i=1}^{K-1}(C_i+C_i'))$. As a consequence of \eqref{negligible} and the last argument, we deduce that for $L_1\geq -\log(\epsilon)+ \log(\sum_{i=1}^{K-1}(C_i+C_i'))$, $Q^T(\bigcup_{i=1}^{K-1} \{|D(t_{i})|>L_1\})\leq \epsilon$. 

By similar arguments, we can prove that for any $j\in \{1, \ldots, K-1\}$, there exists $\hat L_j>0$, such that $Q^{S,j}(\bigcup_{i=1}^{j} \{|D(t_{i})|> \hat L_{j}\})\leq \epsilon$. We can prove also that for  $j\in\{1,\ldots, K-2\}$ and any $i$  such that $t_j< \bar t_i \le t_{j+1}$, and for $j=K-1$ and any $i$ such that $t_{K-1}< \bar t_i \le T$, there exists $L_{i,j}$ such that $Q^{\bar{t}_{i}}(\bigcup_{l=1}^{j} \{|D(t_{l})|> L_{i,j}\}) \leq \epsilon$. Therefore, by defining $\bar L:=\max_{(i,j)} L_{i,j}$ and $L:=\max(L_1,\hat L_1, \ldots, \hat L_{K-1}, \bar{L})$, we deduce that \eqref{bounded} holds true.
\end{proof}

\begin{remark}\label{small constant}
The constant $L$ in Lemma \ref{bislambda} depends on $\epsilon$ and should be denoted $L_{\epsilon}$ but we omit this notation for simplicity. 
\end{remark}

Based on the above, let us replace $W_1$ by another positive and continuous function $W_2$ which coincides with $W_1$ on $[-L,L]$, $L$ as in Lemma \ref{bislambda}, is constant outside this compact set on the positive half line and converges to 0 as $x\to - \infty$. Specifically, we define $W_2$ as
\begin{equation*}
    W_2(x)=\left\{\begin{array}{ll}
        W_1(L)  & \text{for } x\in (L, \infty)\\
         W_1(x) & \text{for } x\in [-L,L]\\
         (\beta L^2+C)e^{L+x} & \text{for } x\in (-\infty,- L)
    \end{array}\right.,
\end{equation*} 
with the constants $ \beta$ and $C$ as in \eqref{lambda}. The new surrender intensity $\bar{\lambda}^s$ is defined as
\begin{equation}
\bar{\lambda}^s(t)= W_2(D(t_i)), \:\: t_i \leq t < t_{i+1}, \label{bar:lambda}
\end{equation}
for $i\in\{1, \ldots K-1\}$,  and $\bar{\lambda}^s(t)=0$ for $t\in [0, t_1) \cup [t_K, T]$.
 Then, as we will see in  Corollary \ref{bis pricing}, the corresponding prices, denoted by ${\rm \bar{P}}^{\rm GMAB}, \bar{{\rm P}}^{\rm DB}$ and  $\bar{{\rm P}}^{\rm SB}$, approximate ${\rm P}^{\rm GMAB}, {\rm P}^{\rm DB}$ and  ${\rm P}^{\rm SB}$. 
 
Starting with the value of the GMAB, we rewrite equation \eqref{A1A2} as
\begin{eqnarray*}
  \frac{{\rm \bar{P}}^{\rm GMAB}}{Q(\tau^m(x) >T)B(0,T)G(T)}&=&  
  E^T\bigg[ e^{-\int_0^{t_K}\bar{\lambda}^s(u)du}\bigg]
  + E^T\bigg[ e^{-\int_0^{t_K}\bar{\lambda}^s(u)du}  \Big(\frac{IS_T}{G(T)}-1\Big)^+ \bigg] \\[2mm]
    &=:& \bar{A}_1 + \bar{A}_2,
  \end{eqnarray*}
  with 
\begin{eqnarray*}
 \bar{A}_1 &=&  E^T\bigg[\prod_{i=2}^K \big(e^{-W_2(D(t_{i-1}))\Delta t_i}\big)\bigg],\\
 \bar{A_2} &=&   E^T\bigg[\prod_{i=2}^K \big(e^{-W_2(D(t_{i-1}))\Delta t_i}\big)\Big(e^{D(T)}-1\Big)^+\bigg].
 \end{eqnarray*}
 
Following the same calculations as in Theorem \ref{PDB}, we deduce that
 \begin{eqnarray*} 
\bar{{\rm P}}^{\rm DB} &=& \sum_{i:\: \bar t_i \le t_1}Q( \tau^m(x) \in [\bar{t}_{i-1}, \bar{t}_i))( G(\bar t_i)B(0, \bar t_i) + G(\bar t_i)B(0, \bar t_i) A_{0,i})\\
&& + \sum_{j=1}^{K-2} \sum_{i:\:\bar t_i\in(t_j,t_{j+1}]} Q\big( \tau^m(x) \in [\bar{t}_{i-1}, \bar{t}_i)\big)G(\bar t_i)B(0,\bar t_i)\big(\bar A^1_{j,i} + \bar A^2_{j,i}\big)\\
&& + \sum_{i:\:\bar t_i\in(t_{K-1},T]} Q\big( \tau^m(x) \in [\bar{t}_{i-1}, \bar{t}_i)\big)G(\bar t_i)B(0,\bar t_i)\big(\bar A^1_{K-1,i} + \bar A^2_{K-1,i}\big),
\end{eqnarray*}
where, $A_{0,i}$ is the same as in Theorem \ref{PDB}, and  for  $j\in\{1,\ldots, K-2\}$ and any $i$  such that $t_j< \bar t_i \le t_{j+1}$, and for $j=K-1$ and any $i$ such that $t_{K-1}< \bar t_i \le T$
 \begin{eqnarray*}
 \bar A^1_{j,i} &:=& E^{\bar t_i}\Big[ e^{-\int_0^{t_{j+1}} \bar \lambda^s(u) du}\Big] = E^{\bar t_i}\Big[ \prod_{i=2}^{j+1}e^{-W_2(D(t_{i-1})\Delta t_i}\Big], \\
 \bar A^2_{j,i} &:=& E^{\bar t_i}\Big[ e^{-\int_0^{t_{j+1}} \bar \lambda^s(u) du}\Big(\frac{IS_{\bar t_i}}{G(\bar t_i)}-1\Big)^+ \Big]\bigg)=E^{\bar t_i}\Big[ \prod_{i=2}^{j+1}e^{-W_2(D(t_{i-1})\Delta t_i}\Big(e^{D_{\bar t_i,\bar t_i} + p(\bar t_i)}-1\Big)^+ \Big].
 \end{eqnarray*}

Finally, by the same argument as in Theorem \ref{surrender}, we deduce that
 \begin{eqnarray*}
    \bar{{\rm P}}^{\rm SB} &=& I \sum_{j=1}^{K-1} P(t_{j}) Q(\tau^m(x) > t_{j}) E_{Q^{S,j}}\bigg[e^{-\int_0^{t_{j}} \bar{\lambda}^s(u)du}  \bigg]\notag\\
    && - I \sum_{j=1}^{K-1} P(t_{j}) Q(\tau^m(x) > t_{j}) E_{Q^{S,j}}\bigg[e^{-\int_0^{t_{j+1}} \bar{\lambda}^s(u)du}  \bigg]\notag\\
    &=& I \sum_{j=1}^{K-1} P(t_{j}) Q(\tau^m(x) > t_{j})(\bar{B}_j^{1} - \bar{B}_j^{2}),
  \end{eqnarray*}
 with $\bar{B}_1^{1}=1$, due to $\bar{\lambda}^s(u)=0$, for $u\in[0, t_1)$ by construction, and
 \begin{eqnarray*}
 \bar{B}_j^{1} &=&  E_{Q^{S,j}}\bigg[\prod_{l=2}^{j} \big(e^{-  W_2(D(t_{l-1}))\Delta t_l}\big)\bigg], \text{ for } j\geq 2, \notag\\ 
 \bar{B}_j^{2}&=&  E_{Q^{S,j}}\bigg[\prod_{l=2}^{j+1} \big(e^{- W_2(D(t_{l-1}))\Delta t_l}\big)\bigg].
 \end{eqnarray*}
 
Given the above, the following result holds.
\begin{proposition} \label{Prop:barA1} We have
\begin{itemize}
    \item[(1)] $|\bar{A}_1 - A_1| \leq 2 \epsilon$,\\
    \item[(2)] $|\bar{A}_2 -A_2|\leq 2 C_2  \epsilon^{1/2}$, with $C_2:=E^T\big[ (e^{D(T)}-1)^2\big]^{1/2}$,\\
    \item[(3)] $|\bar{A}^1_{j,i} - A^1_{j,i}| \leq 2 \epsilon$,\\
    \item[(4)] $|\bar{A}^2_{j,i} -A^2_{j,i}|\leq 2 C_{2,i}  \epsilon^{1/2}$, with $C_{2,i}:=E^{\bar t_i}\big[ (e^{D_{\bar t_i, \bar t_i} +p(\bar t_i)}-1)^2\big]^{1/2}$,\\
   \item[] where in  (3) and (4) we have either  $j\in\{1,\ldots, K-2\}$ and  $i$  such that $t_j< \bar t_i \le t_{j+1}$, or $j=K-1$ and  $i$ such that $t_{K-1}< \bar t_i \le T$,\\
    \item[(5)] $|\bar{B}_j^{1} - B_j^{1}| \leq 2 \epsilon $, for $j \in\{2, \ldots, K-1\}$,\\
    \item[(6)] $|\bar{B}_j^{2} - B_j^{2}| \leq 2 \epsilon$, for $j \in\{1, \ldots, K-1\}$.
    \end{itemize}
$A_1$, $A_2$ are given in \eqref{A1A2}, $A^1_{j,i}$, $A^2_{j,i}$ in \eqref{A1A2'}, and $B_j^{1}$, $B_j^{2}$ are given in \eqref{B^1} and \eqref{B^2}.
\end{proposition}

\begin{proof}
It suffices to prove (2) which is representative for the degree of sophistication. The proofs of (1), (3), (4), (5) and (6) follow mutatis mutandis. From the definition of $W_1$ given in \eqref{W1}, we derive 
\begin{eqnarray*}
  \bar{A}_2&=& E^T\bigg[\prod_{i=2}^K \big(e^{-W_2(D(t_{i-1}))\Delta t_i}\big)\Big(e^{D(T)}-1\Big)^+\bigg]\\
& =& E^T\bigg[\prod_{i=2}^K \big(e^{-W_2(D(t_{i-1}))\Delta t_i}\big)\Big(e^{D(T)}-1\Big)^+ \ind{\cap_{i=2}^{K} \{|D(t_{i-1})| \leq L\}}\bigg]\\
&& + \; E^T\bigg[\prod_{i=2}^K \big(e^{-W_2(D(t_{i-1}))\Delta t_i}\big)\Big(e^{D(T)}-1\Big)^+ \ind{\cup_{i=2}^{K} \{|D(t_{i-1})| > L\}}\bigg]\\
&=& E^T\bigg[\prod_{i=2}^K \big(e^{-W_1(D(t_{i-1}))\Delta t_i}\big)\Big(e^{D(T)}-1\Big)^+ \ind{\cap_{i=2}^{K} \{|D(t_{i-1})| \leq L\}}\bigg]\\
&& + \; a_1,
\end{eqnarray*}
with the obvious definition of $a_1$ in the last line. The third equality holds as $W_1$ and $W_2$ coincide on $[-L, L]$. Observe that
\begin{eqnarray*}
&& E^T\bigg[\prod_{i=2}^K \big(e^{-W_1(D(t_{i-1}))\Delta t_i}\big)\Big(e^{D(T)}-1\Big)^+ \ind{\cap_{i=2}^{K} \{|D(t_{i-1})| \leq L\}}\bigg]\\
&&= E^T\bigg[\prod_{i=2}^K \big(e^{-W_1(D(t_{i-1}))\Delta t_i}\big)\Big(e^{D(T)}-1\Big)^+\bigg]\\
&&\quad -\; E^T\bigg[\prod_{i=2}^K \big(e^{-W_1(D(t_{i-1}))\Delta t_i}\big)\Big(e^{D(T)}-1\Big)^+ \ind{\cup_{i=2}^{K} \{|D(t_{i-1})| > L\}}\bigg]\\
&&=  A_2 - a_2,
\end{eqnarray*}
with the obvious notation in the last line. We deduce that
\[
|\bar{A}_2 - A_2| \leq (a_1 +a_2).
\]
As $W_1$ is a positive function, the following holds
\begin{eqnarray*}
 a_2 &\leq& E^T\bigg[\Big(e^{D(T)}-1\Big)^+ \ind{\cup_{i=1}^{K-1} \{|D(t_{i})| > L\}}\bigg]\\
&\leq& E^T\bigg[\Big|e^{D(T)}-1\Big| \ind{\cup_{i=1}^{K-1} \{|D(t_{i})| > L\}}\bigg]\\
&\leq& E^T\bigg[ (e^{D(T)}-1)^2\bigg]^{1/2}Q^T(\bigcup_{i=1}^{K-1} \{|D(t_{i})|>L\})^{1/2}\\
& \leq& C_2 \epsilon^{1/2},
\end{eqnarray*}
where the third line is a consequence of the Cauchy-Schwarz inequality. For the last inequality we use \eqref{bounded} together with the fact that $C_2:=E^T\bigg[ (e^{D(T)}-1)^2\bigg]^{1/2}< \infty$. This last fact is easy to verify from the definition of $D(T)$ in \eqref{D:appendix}. A similar bound can be achieved for $a_1$, therefore 
\[
|\bar{A}_2 - A_2| \leq 2  C_2 \epsilon^{1/2},
\]
which proves (2). 
\end{proof}

Based on the last proposition, the interested reader can prove the following corollary.
\begin{corollary}\label{bis pricing} We have

\begin{itemize}
    \item[(1)] $|{\rm \bar{P}}^{\rm GMAB} -{\rm P}^{\rm GMAB}|\leq 2  Q(\tau^m(x) >T)B(0,T)G(T)[\epsilon + C_2 \epsilon^{1/2}],$ \\
    with $C_2:=E^T\big[ (e^{D(T)}-1)^2\big]^{1/2}.$\\
    
    \item[(2)] $|{\rm \bar{P}}^{\rm DB} -{\rm P}^{\rm DB}|\leq 2\sum_{i:\: \bar t_i > t_1} Q\big( \tau^m(x) \in [\bar{t}_{i-1}, \bar{t}_i)\big)B(0,\bar{t}_i)G(\bar{t}_i)[\epsilon + C_{2,i} \epsilon^{1/2}]$,\\
    with $C_{2,i}:= E^{\bar{t}_i}\big[ (e^{D_{\bar{t}_i,\bar{t}_i}   +p(\bar{t}_i)}-1)^2\big]^{1/2}$.\\
    
    \item[(3)] $|{\rm \bar{P}}^{\rm SB} -{\rm P}^{\rm SB}|\leq 2  I P(t_1)Q(\tau^m(x) > t_{1})\epsilon +4 I \epsilon \sum_{j=2}^{K-1} P(t_j) Q(\tau^m(x) > t_{j}).$
\end{itemize}
\end{corollary}

\section{Some useful results and representations}\label{app:prelim}
 
 Let us recall that $A(u,T)= \int_u^T \alpha(u,s)ds$ and $\Sigma(u,T)= \int_u^T \sigma_1(u,s)ds$. 
 We derive here a representation result in the L\'evy forward rate framework.
  \begin{lemma} For any $0\leq t\leq T$, we have that\label{lem:HJM}
  $$
    -\int_{t}^T f(t,s)ds = \int_0^{t}r(s)ds -\int_0^T f(0,s)ds - \int_0^{t}A(u,T)du + \int_0^{t}\Sigma(u,T)dL^1_u,
  $$
  where $f(t,s)$ is defined in \eqref{forward rate}.
  \end{lemma}
 
 \begin{proof}
Using Fubini's theorem for stochastic integrals, we deduce that
 \begin{eqnarray*}
 \lefteqn{ -\int_{t}^T f(t,s)ds = \int_0^{t}r(s)ds -\int_0^{t} r(s)ds - \int_{t}^T f(0,s)ds - \int_{t}^T \int_0^{t}\alpha(u,s)du\, ds} \hspace{5mm}\\
 && +\int_{t}^T\int_0^{t}\sigma_1(u,s)dL^1_uds\\
 &=& \int_0^{t}r(s)ds -\Big[\int_0^{t}f(0,s)ds + \int_0^{t}(\int_0^s \alpha(u,s)du)ds - \int_0^{t}(\int_0^s \sigma_1(u,s)dL^1_u)ds\Big] \\
 && - \int_{t}^T f(0,s)ds - \int_0^{t}\int_{t}^T \alpha(u,s)ds\, du +\int_0^{t}\int_{t}^T\sigma_1(u,s) ds\,dL^1_u\\
 &=& \int_0^{t}r(s)ds  -\Big[\int_0^{t}f(0,s)ds + \int_0^{t}\int_u^{t} \alpha(u,s)ds\,du - \int_0^{t}\int_u^{t} \sigma_1(u,s)ds\,dL^1_u\Big] \\
 && - \int_{t}^T f(0,s)ds - \int_0^{t}\int_{t}^T \alpha(u,s)ds\,du +\int_0^{t}\int_{t}^T\sigma_1(u,s) ds\,dL^1_u,
 \end{eqnarray*}
 and the claim follows.
  \end{proof}

The above lemma allows the following representation of $D(t)$ defined in \eqref{Dt}.
\begin{align}
     D(t)&= \int_0^t r(s)ds + \int_0^t \sigma_2(s)dL^2_s + \int_0^t\beta(s)dL^1_s -\omega(t) -p(t) + \int_{t}^Tf(t,s)ds - \delta T  \notag \\
     &=  -p(t) - \delta T+ \int_0^T f(0,s)ds + \int_0^{t}A(s,T)ds - \int_0^{t}\Sigma(s,T)dL^1_s \label{D:appendix}\\
     & \quad +\int_0^{t} \sigma_2(s)dL^2_s + \int_0^{t}\beta(s)dL^1_s -\omega(t). \notag
 \end{align}

In addition, setting $t=T$ in the above lemma, we obtain  
  \begin{equation}
    0 = \int_0^{t}r(s)ds -\int_0^t f(0,s)ds - \int_0^{t}A(s,t)ds + \int_0^{t}\Sigma(s,t)dL^1_s \label{interest}
  \end{equation}
and hence, the bank account has the following well-known representation, see for example \cite{EberleinRaible99} or (13) in \cite{ER},
  \begin{equation}
  B(t) = \frac{1}{B(0,t)}\exp\Big( \int_0^t A(s,t)ds - \int_0^t \Sigma(s,t)dL^1_s \Big). \label{Bank ac}
  \end{equation}  

\section{Proof of Theorem \ref{PDB}}\label{app:proofDB}

\begin{proof}
  By definition, 
  \begin{align*}
      {\rm P}^{\rm DB} &= \sum_{i=1}^N  E_{Q}\bigg[e^{-\int_0^{\bar t_i} r(u)du} {\rm DB}(\bar t_i)\bigg].
  \end{align*}
  From equation \eqref{mortality indep}, we obtain that
  \begin{align*}
    E_{Q}\bigg[e^{-\int_0^{\bar t_i} r(u)du} {\rm DB}(\bar t_i)\bigg]
     &=  Q( \tau^m(x) \in [\bar{t}_{i-1}, \bar{t}_i))E_{Q}\bigg[e^{-\int_0^{\bar{t}_i} r(u)du} \ind{ \tau^s \geq \bar{t}_i}\max(I S(\bar{t}_i), G(\bar{t}_i))\bigg]. 
  \end{align*}
  We distinguish two cases: when $\bar t_i \le t_1$ and when $t_1 < \bar t_i$. We start with the detailed description of the second case which is the more involving one. The first case is treated at the end.
    
For $j\in \{1,\ldots, K-2\}$ and $i$ such that $t_j < \bar t_i \le t_{j+1}$, as well as for $j=K-1$ and  $i$ such that $t_{K-1} < \bar t_i\leq T$, we work along the same line as in the proof of  Theorem \ref{prop:A1}, and get 
\begin{eqnarray}
  && E_{Q}\bigg[e^{-\int_0^{\bar{t}_i} r(u)du} \ind{ \tau^s \geq \bar{t}_i}\max(I S(\bar{t}_i), G(\bar{t}_i))\bigg]\notag\\
  &&=E_{Q}\bigg[e^{-\int_0^{\bar{t}_i} r(u)du} \ind{ \tau^s \geq t_{j+1}}\max(I S(\bar{t}_i), G(\bar{t}_i))\bigg]\notag\\
  &&=E_{Q}\bigg[e^{-\int_0^{\bar{t}_i} r(u)du} e^{-\int_0^{t_{j+1}} \lambda^s(u) du}\max(I S(\bar{t}_i), G(\bar{t}_i))\bigg]\notag\\
  &&= G(\bar t_i)E_{Q}\bigg[e^{-\int_0^{\bar{t}_i} r(u)du} e^{-\int_0^{t_{j+1}} \lambda^s(u) du}\Big
  (1+ \Big(\frac{IS_{\bar t_i}}{G(\bar t_i)}-1\Big)^+\Big) \bigg].    \label{DB1}
  \end{eqnarray}
  We introduce the $\bar t_i$-forward measure $Q^{\bar t_i}$ defined by its Radon-Nikodym density
  \begin{equation}
  \frac{dQ^{\bar t_i}}{dQ}= \frac{1}{B(0,\bar t_i)B(\bar t_i)}. \label{second forward measure}
  \end{equation}
  Denoting the expectation with respect to $Q^{\bar t_i}$ by $E^{\bar t_i}$ the quantity in equation \eqref{DB1} is 
  \begin{eqnarray}
  &&G(\bar t_i)B(0,\bar t_i)\bigg( E^{\bar t_i}\Big[ e^{-\int_0^{t_{j+1}} \lambda^s(u) du}\Big] + E^{\bar t_i}\Big[ e^{-\int_0^{t_{j+1}} \lambda^s(u) du}\Big(\frac{IS_{\bar t_i}}{G(\bar t_i)}-1\Big)^+ \Big]\bigg)\notag\\
  &&= \ G(\bar t_i)B(0,\bar t_i)\big( A_{j,i}^1 + A_{j,i}^2\big), \label{A1A2'}
  \end{eqnarray}  with an obvious notation in the last line. We note that
  \begin{equation}
 e^{C(t_{j+1} -t_1)} A_{j,i}^1= E^{\bar t_i}\bigg[f(D(t_1), \ldots, D(t_{j}))\bigg], \label{Aji1}
 \end{equation}
with
  $f(x_1, \ldots, x_{j})
              =\prod_{l=2}^{j+1}  e^{-\beta \Delta t_{l} x_{l-1}^2}$, and $\Delta t_l =t_l -t_{l-1}.$ 
As in \eqref{Price-Fourier} we obtain that the expectation in \eqref{Aji1} can be represented as 
  \begin{equation}
    \frac{1}{(2\pi)^{j}}\int_{\mathbb{R}^{j}}\tilde M^{j}_i(\i u)\hat{f}(-u)du, \label{DB fourier1}
  \end{equation}
  with
  \begin{equation*}
    \hat{f}(u_1, \ldots, u_{j})=  \prod_{l=2}^{j+1} \sqrt{\frac{\pi}{\beta \Delta t_{l}}}e^{-u_{l-1}^2/(4 \beta \Delta t_{l})},
  \end{equation*}  
  and $\tilde M^{j}_i(\i u)$  defined as 
    \begin{align*}
       \tilde M^{j}_i(\i u) &= E^{\bar t_i}\bigg[e^{\i u_1 D(t_1) + \ldots + \i u_{j}D(t_{j})}\bigg].
    \end{align*}
By the representation of the bank account in \eqref{Bank ac},  the Radon-Nikodym density \eqref{second forward measure} can be written as 
  \begin{equation} 
    \frac{dQ^{\bar t_i}}{dQ}= \exp\Big( -\int_0^{\bar t_i} A(s,\bar t_i)ds + \int_0^{\bar t_i} \Sigma(s,\bar t_i)dL^1_s \Big). \label{second forward measure'}
  \end{equation}
Consequently, from the representation of $D(t)$ in \eqref{D:appendix}, it follows
\begin{eqnarray*}
\tilde M^{j}_i(\i u) &=& E^{\bar t_i}\bigg[e^{\i u_1 D(t_1) + \ldots + \i u_{j}D(t_{j})}\bigg]\\
&=& E_Q\bigg[e^{\i u_1 D(t_1) + \ldots + \i u_{j}D(t_{j})-\int_0^{\bar t_i} A(s,\bar t_i)ds + \int_0^{\bar t_i} \Sigma(s,\bar t_i)dL^1_s}\bigg]\\
&=& \exp\Big( \i\sum_{l=1}^{j} u_{l}\big(-p(t_{l}) - \delta T+ \int_0^T f(0,s)ds + \int_0^{t_{l}}A(s,T)ds -\omega(t_{l})\big)-\int_0^{\bar t_i} A(s,\bar t_i)ds   \Big)\\
&& \times E_{Q}\bigg[
       \exp\bigg( \i\sum_{l=1}^{j} \Big(\int_0^{t_{l}} u_{l}\sigma_2(s)dL^2_s + \int_0^{t_{l}}u_{l}(\beta(s) - \Sigma(s,T))dL^1_s\Big) + \int_0^{\bar t_i} \Sigma(s,\bar t_i)dL^1_s \bigg)   \bigg]. 
\end{eqnarray*}
This last expectation is finite in virtue of \eqref{AssM}, and returns
  \begin{align*}
      \lefteqn{ E_{Q}\Big[\exp\Big(  \int_0^{\bar t_{i}} E_{j,i}(s,u,T) dL^1_s + \int_0^{\bar t_{i}} F_{j}(s,u) dL^2_s \Big) \Big] } \qquad \qquad \\ 
      &= \exp\Big(  \int_0^{\bar t_{i}} \Big( \theta^1_s( E_{j,i}(s,u,T) )+ \theta^2_s ( F_{j}(s,u)) \Big) ds  \Big), 
  \end{align*}
  in virtue of the definitions in \eqref{DB notation} and equation \eqref{cumulant}.
  Therefore,  with $D^{j,i}(u,T)$ defined in \eqref{DB notation} we have
  \begin{equation*}
    \tilde M^{j}_i(\i u) = D^{j,i}(u,T)  \exp\Big(  \int_0^{\bar t_{i}} \Big( \theta^1_s( E_{j,i}(s,u,T) )+ \theta^2_s ( F_{j}(s,u)) \Big) ds  \Big). 
  \end{equation*}
Finally, combining \eqref{Aji1} with \eqref{DB fourier1} and the definition of $M^{j,i}(u,T)$ in \eqref{DB notation}, we deduce that 
\[
A_{j,i}^1= \frac{e^{-C(t_{j+1} -t_1)}}{(2\pi)^{j}}\int_{\mathbb{R}^{j}}M^{j,i}(u,T) du.
\]

For the purpose of the computation of $A_{j,i}^2$, we replace $D(t)$ defined in \eqref{Dt} with the following quantity
  \begin{equation}
  D_{t,t'}= Y_{t} -p(t) + \int_{t}^{t'}f(t,s)ds - \delta t', \quad 0 \leq  t \leq t', \label{bis:Dt}
  \end{equation}
so that 
  \[
  \Big(\frac{IS_{\bar t_i}}{G(\bar t_i)}-1\Big)^+ = \Big(\exp\big(D_{\bar t_i, \bar t_i} +p(\bar t_i)\big) -1\Big)^+.
  \]
We note that  
 \begin{equation}
 e^{C(t_{j+1} -t_1)}A_{j,i}^2= E^{\bar t_i}\big[h\big(D(t_1), \ldots, D(t_{j}), D_{\bar t_i, \bar t_i}\big)\big], \label{Aji2}
\end{equation}
for $h(x_1, \ldots, x_{j+1}):=f(x_1,\dots,x_{j})  (e^{x_{j+1} +p(\bar t_i) }-1)^+ ,$  with $f$ given in \eqref{Aji1}.
In order to ensure integrability,  let us define $H(x_1,\ldots, x_{j+1}):=h(x_1, \ldots, x_{j+1}) e^{-rx_{j+1}}$, for some $1<r<2$, and 
$$ H_{j+1} (x_{j+1}) :=  (e^{x_{j+1} +p(\bar t_i) }-1)^+ e^{-r x_{j+1}}. $$
Then, $H_{j+1} \in L^1(\mathbb{R})$, and $H\in L^1(\mathbb{R}^{j+1})$.
Moreover, elementary integration shows that
for all $ y \in \mathbb{R}$
$$ \hat{H}_{j+1}(y)= \frac{\exp\big(-p(\bar t_i)(\i y-r)\big)}{(\i y -r+1)(\i y-r)}. $$ 
Observe that $|\hat{H}_{j+1}(y)|_{\mathbb{C}}= e^{rp(\bar t_i)}(((1-r)^2+ y^2)(r^2+y^2))^{-1/2}$, thus,
 $\hat{H}_{j+1}\in L^1(\mathbb{R})$. Therefore, combining the last result with \eqref{eq:fhat}, we deduce that $\hat{H}\in L^1(\mathbb{R}^{j+1})$, and
 \begin{equation}
   \hat{H}(y_1, \ldots, y_{j+1})= \frac{\exp\big(-p(\bar t_i)(\i y_{j+1} -r)\big)}{(\i y_{j+1} -r+1)(\i y_{j+1}-r)} \prod_{l=2}^{j+1} \sqrt{\frac{\pi}{\beta \Delta t_l}}e^{-y_{l-1}^2/(4 \beta \Delta t_l)}. \label{DB fourier2}
 \end{equation}
 As $H, \hat{H}\in L^1(\mathbb{R}^{j+1})$, it follows from Theorem 3.2 in \cite{EP} that
 \begin{equation}
E^{\bar t_i}\big[h\big(D(t_1), \ldots, D(t_{j}), D_{\bar t_i, \bar t_i}\big)\big] = \frac{1}{(2\pi)^{j+1}}\int_{\mathbb{R}^{j+1}}\tilde N_i^{j+1}(R+\i u)\hat{h}(\i R-u)du, \label{DB fourier3}
\end{equation}
  for $R=(0, \ldots, 0, r)\in\mathbb{R}^{j+1}$, $1<r<2$, and  $\tilde N_i^{j+1}(R+\i u)$ defined as 
  \begin{equation}
   \tilde N_i^{j+1}(R+\i u):= E^{\bar t_i}\big[e^{\i u_1 D(t_1) + \ldots + \i u_{j}D(t_{j})+ (\i u_{j+1} +r)D_{\bar t_i, \bar t_i}}\big]. \label{DB tildeN}
  \end{equation}
  Using \eqref{bis:Dt} and \eqref{interest}, we get
  \begin{eqnarray*}
  D_{\bar t_i, \bar t_i}&=& -p(\bar t_i) - \delta \bar t_i + Y_{\bar t_i}\\
  &=& -p(\bar t_i) - \delta \bar t_i +\int_0^{\bar t_i} r(s)ds +\int_0^{\bar t_i} \sigma_2(s)dL^2_s + \int_0^{\bar t_i}\beta(s)dL^1_s -\omega(\bar t_i)\\
  &=& -p(\bar t_i) - \delta \bar t_i + \int_0^{\bar t_i} f(0,s)ds + \int_0^{\bar t_i}A(s,\bar t_i)ds - \int_0^{\bar t_i}\Sigma(s,\bar t_i)dL^1_s\\
  && +\int_0^{\bar t_i} \sigma_2(s)dL^2_s + \int_0^{\bar t_i}\beta(s)dL^1_s -\omega(\bar t_i)\\
  &=& -p(\bar t_i) - \delta \bar t_i -\omega(\bar t_i)+ \int_0^{\bar t_i} f(0,s)ds + \int_0^{\bar t_i}A(s,\bar t_i)ds + \int_0^{\bar t_i} \sigma_2(s)dL^2_s\\
  && + \int_0^{\bar t_i}\big(\beta(s)-\Sigma(s,\bar t_i)\big)dL^1_s.
  \end{eqnarray*}
  Plugging the last quantity in \eqref{DB tildeN}, in virtue of  \eqref{D:appendix} and \eqref{second forward measure'}, we deduce that 
  \begin{eqnarray*}
  &&\tilde N_i^{j+1}(R+\i u)\\&&= E_Q\bigg[e^{\i u_1 D(t_1) + \ldots + \i u_{j}D(t_{j}) +(\i u_{j+1} +r)D_{\bar t_i, \bar t_i}-\int_0^{\bar t_i} A(s,\bar t_i)ds + \int_0^{\bar t_i} \Sigma(s,\bar t_i)dL^1_s}\bigg]\\
&&= \exp\Big( \i\sum_{l=1}^{j} u_{l}\big(-p(t_{l}) - \delta T+ \int_0^T f(0,s)ds + \int_0^{t_{l}}A(s,T)ds -\omega(t_{l})\big)-\int_0^{\bar t_i} A(s,\bar t_i)ds   \Big)\\
&&\quad\times \exp\Big((\i u_{j+1} +r) \big(-p(\bar t_i) - \delta \bar t_i -\omega(\bar t_i)+ \int_0^{\bar t_i} f(0,s)ds + \int_0^{\bar t_i}A(s,\bar t_i)ds\big) \Big)\\
&&\quad  \times E_{Q}\bigg[
       \exp\bigg( \i\sum_{l=1}^{j} \Big(\int_0^{t_{l}} u_{l}\sigma_2(s)dL^2_s + \int_0^{t_{l}}u_{l}(\beta(s) - \Sigma(s,T))dL^1_s\Big)\\ &&\qquad +\;\i \Big(\int_0^{\bar t_{i}} u_{j+1}\sigma_2(s)dL^2_s 
        + \int_0^{\bar t_{i}}u_{j+1}(\beta(s) - \Sigma(s,\bar t_i))dL^1_s \Big)\\ 
        &&\qquad +  \int_0^{\bar t_{i}} r\sigma_2(s)dL^2_s + \int_0^{\bar t_{i}} r(\beta(s) - \Sigma(s,\bar t_i))dL^1_s 
        +\int_0^{\bar t_i} \Sigma(s,\bar t_i)dL^1_s \bigg)   \bigg].
  \end{eqnarray*}
Using the definitions from \eqref{DB notation}, the above can be rewritten as
\begin{align*}
\tilde N^{j+1}_i(R+\i u)
  &= \tilde D^{j,i}(u -\i R ,T)
      \; E_Q\bigg[\exp\Big( \int_0^{\bar t_i} \tilde E_{j,i}(s,u-\i R,T) dL^1_s + \int_0^{\bar t_i} \tilde F_j(s,u-\i R) dL^2_s \Big)\bigg].
  \end{align*}
 Observe that, due to $1<r <2$, as well as \eqref{AssM} and \eqref{assM12},  $r \sigma_2(s)\leq M_2$ and $|r\beta(s) +(1-r)\Sigma(s,T)| \leq (2r-1)\frac{M_1}{3} \leq M_1$. Thus, the above expectation exists.
 Using the independence of $L^1$ and $L^2$ and \eqref{cumulant}, we obtain that\small
 \begin{eqnarray}
 \tilde N_i^{j+1}(R+\i u)
  &=& \tilde D^{j,i}(u-\i R  ,T) \notag\\
  &&\times \exp\Big( \int_0^{\bar t_i} \theta^1_s(\tilde E_{j,i}(s,u-\i R,T)) ds + \int_0^{\bar t_i} \theta^2_s(\tilde{F_j}(s,u-\i R))    ds \Big). \label{DB N'}
 \end{eqnarray}\normalsize
 On the other hand, we observe that for any $u\in \mathbb{R}^{j+1}$, 
 $$\hat{H}(u)=\int_{\mathbb{R}^{j+1}}e^{\i \langle u, x \rangle } e^{- \langle R , x \rangle} h(x)dx = \hat{h}(u+\i R). $$ Consequently, we deduce that
 \begin{equation}
 \hat{h}(\i R - u) =\hat{H}(-u) 
 = \frac{\exp\big(p(\bar t_i)(\i u_{j+1} +r)\big)}{(\i u_{j+1} +r-1)(\i u_{j+1}+r)} \prod_{l=2}^{j+1} \sqrt{\frac{\pi}{\beta \Delta t_l}}e^{-u_{l-1}^2/(4 \beta \Delta t_l)}. \label{DB fourier4}
 \end{equation}
Plugging \eqref{DB N'} and \eqref{DB fourier4} in \eqref{DB fourier3}, and using the definition of $ N^{j,i}(u,T)$ in \eqref{DB notation}, it follows that 
\[
A_{j,i}^2= \frac{e^{-C(t_{j+1} -t_1)}}{(2\pi)^{j+1}}\int_{\mathbb{R}^{j+1}}N^{j,i}(u,T) du.
\]

We now consider the case $\bar t_i \le t_1$ for $i\in \{1, \ldots, N\}$. Using the same arguments as above, we derive
\begin{eqnarray*}
&& E_{Q}\bigg[e^{-\int_0^{\bar{t}_i} r(u)du} \ind{ \tau^s \geq \bar{t}_i}\max(I S(\bar{t}_i), G(\bar{t}_i))\bigg]\\
&&= E_{Q}\bigg[e^{-\int_0^{\bar{t}_i} r(u)du} \max(I S(\bar{t}_i), G(\bar{t}_i))\bigg]\\
&& = B(0, \bar t_i) E^{\bar t_i}\bigg[\max(I S(\bar{t}_i), G(\bar{t}_i))\bigg]\\
&&= B(0, \bar t_i)G(\bar t_i) E^{\bar t_i}\bigg[(1+ \Big(\frac{IS_{\bar t_i}}{G(\bar t_i)}-1\Big)^+\Big)\bigg]\\
&&= B(0, \bar t_i)G(\bar t_i) + B(0, \bar t_i)G(\bar t_i) E^{\bar t_i}\bigg[\Big(\exp\big(D_{\bar t_i, \bar t_i} +p(\bar t_i)\big) -1\Big)^+\bigg]\\
&& = B(0, \bar t_i)G(\bar t_i) + B(0, \bar t_i)G(\bar t_i) E^{\bar t_i}\big[h_1(D_{\bar t_i, \bar t_i})\big],
\end{eqnarray*}
where $h_1(x):=(e^{x +p(\bar t_i) }-1)^+$. For some $1<r<2$, we define the function $H_1$ as 
 $H_1 (x) :=  (e^{x +p(\bar t_i) }-1)^+ e^{-r x}.$ By Theorem 3.2 in \cite{EP}, we get  
 \begin{equation}
E^{\bar t_i}\big[h_1\big(D_{\bar t_i, \bar t_i}\big)\big] = \frac{1}{2\pi}\int_{\mathbb{R}} \tilde N_i(r+\i u)\hat{h}_1(\i r-u)du, \label{DB fourier5}
\end{equation}
with 
\[
\tilde N_i(r+\i u):=E^{\bar t_i}\big[e^{(r+ \i u )D_{\bar t_i, \bar t_i}}\big].
\]
Using the same arguments as above, and the definitions in \eqref{DB notation}, we deduce that
\begin{eqnarray}
&&\tilde N_i(r+\i u)\notag\\
&&= \exp\Big((r +\i u)w_{\bar t_i} -\int_0^{\bar t_i} A(s,\bar t_i)ds \Big)\; E_Q\Big[\exp\Big(\int_0^{\bar t_i}E_1(s,u)dL^1_s + \int_0^{\bar t_i}F_1(s,u)dL^2_s\Big)\Big]\notag\\
&&= \exp\Big((r +\i u)w_{\bar t_i} -\int_0^{\bar t_i} A(s,\bar t_i)ds \Big)\;\exp\Big( \int_0^{\bar t_i}\theta^1_s(E_1(s,u))ds + \int_0^{\bar t_i}\theta^2_s(F_1(s,u))ds\Big).\notag\\ \label{DB Ni}
\end{eqnarray}
On the other hand, we have
\begin{equation}
 \hat{h}_1(\i r - u) =\hat{H}_1(-u) 
 = \frac{\exp\big(p(\bar t_i)(\i u +r)\big)}{(\i u +r-1)(\i u+r)}. \label{DB fourier6}
 \end{equation}
 Plugging \eqref{DB Ni} and \eqref{DB fourier6} in \eqref{DB fourier5},  we get
 \[
 E^{\bar t_i}\big[h_1\big(D_{\bar t_i, \bar t_i}\big)\big]=\frac{1}{2\pi}e^{-\int_0^{\bar t_i} A(s,\bar t_i)ds}\int_{\mathbb{R}} N^i(u)du.
 \]
 Therefore, 
\begin{equation*}
     E_{Q}\bigg[e^{-\int_0^{\bar{t}_i} r(u)du} \ind{ \tau^s \geq \bar{t}_i}\max(I S(\bar{t}_i), G(\bar{t}_i))\bigg] =B(0, \bar t_i)G(\bar t_i) + B(0, \bar t_i)G(\bar t_i)A_{0,i}.
\end{equation*} 

Finally, we mention that for any $1\leq i \leq N$, we can compute $Q( \tau^m(x)> \bar{t}_{i} )$ (and consequently $Q( \tau^m(x) \in [\bar{t}_{i-1}, \bar{t}_i))$) in virtue of \eqref{Q^m}.

 \end{proof}
 
 \section{Further numerical results: benchmarking}\label{App:bench}
 In this Appendix, we provide full results from the benchmarking exercise of the Monte Carlo integration pricing procedure. 
 
 From section \ref{sec:numerics}, we recall that we consider contracts with short maturity, and surrender frequency $\Delta t_l=1$ year. In addition, we assume half-annually spaced $\bar{t}_i$ for the mortality monitoring. 
 
 Consequently, a sensible benchmarking exercise for the value of both the DB and the SB with deterministic quadrature methods can be achieved by considering the 4 year maturity contract, i.e. $K=3$. 
 
 Starting with the value of the DB,  the term $A_{0,i}$ in Theorem \ref{PDB} is a 1-dimensional integral for $i=1,\dots,N$, which can be obtained by direct quadrature, and therefore is not considered by this benchmarking exercise. 
 
 Monte Carlo integration is instead used to compute the remaining terms appearing in Theorem \ref{PDB}, i.e. $A_{1,i}^{1}$, $A_{1,i}^{2}$, $A_{2,i}^{1}$, which are two-dimensional integrals, and $A_{2,i}^{2}$ which is a three-dimensional integral, for all $i=1,\dots,N$. These are benchmarked against the values obtained by deterministic quadrature procedures in Matlab.

For importance sampling, we choose the same values of the variance of the importance sampling distribution as in the GMAB case, i.e. 0.25 for the first $K-1$ dimensions and 1 for the final $K^{th}$ dimension. The quality of the estimate is confirmed by the negligible bias and standard errors reported in Table \ref{tab:GMABbench_f}.

Concerning the value of the SB, we notice that the first term in the sum defining $P^{SB}$ in Theorem \ref{surrender} is composed by a constant ($B_1^1$) and a one-dimensional integral ($B_1^2$), which is obtained by deterministic  quadrature; similarly to the previous case this term is not considered in this analysis. The second term in this sum is instead formed by a one- and a two-dimensional integral ($B_2^1$ and $B_2^2$ respectively) for which we deploy Monte Carlo integration. Benchmarking is performed against the corresponding quadrature routines in Matlab.

For importance sampling, given the relatively simple forms of the integrand functions, we use the same variance fixed at 0.16 across all the $K$ dimensions. The goodness of the estimate is confirmed by the negligible bias and standard errors shown in Table \ref{tab:GMABbench_f}.
 
\begin{table}[h]
  \centering
  \caption{\small Benchmarking Monte Carlo integration with importance sampling. Parameters: Table \ref{tab:NIGParams}. `Quadrature':  Matlab built-in functions {\fontfamily{qcr}\selectfont integral}, {\fontfamily{qcr}\selectfont integral2} and {\fontfamily{qcr}\selectfont integral3}. Bias/standard error expressed as percentage of the actual value. Monte Carlo iterations: 100 batches of size $10^6$. CPU time expressed in seconds and referred to the average time of 1 batch of $10^6$ iterations.}
  \resizebox {0.75\textwidth }{!}{
    \begin{tabular}{lclccccr}\toprule
      \multicolumn{1}{l}{GMAB} &       &       & Quadrature & \multicolumn{4}{c}{Monte Carlo integration (Imp. Sampling)} \\
      \multicolumn{1}{c}{$T$} & $K$     &       & Value & Value & Bias (\%) & Std. Error (\%) &  \\\midrule
      \multicolumn{1}{l}{3 years} & 2     & $A_1$  & 0.9867 & 0.9867 & 0.0035 & 0.0050 &  \\
      &       & $A_2$  & 0.1487 & 0.1482 & 0.3584 & 0.2683 &  \\
      &       & CPU   & 6.6323 & 31.2944 & &    &  \\
      &       &       & \multicolumn{1}{l}{($A_1$: 0.1280)} &       &       &       &  \\
      &       &       & \multicolumn{1}{l}{($A_2$: 6.5043)} &       &       &       &  \\
      &       &       &       &       &       &       &  \\
      \multicolumn{1}{l}{4 years} & 3     & $A_1$  & 0.9703 & 0.9702 & 0.0139 & 0.0076 &  \\
      &       & $A_2$  & 0.1669 & 0.1669 & 0.0155 & 0.0647 &  \\
      &       & CPU   & 589.0926 & 51.6269 & &  &  \\
      &       &       & \multicolumn{1}{l}{($A_1$: 1.3042)} &       &       &       &  \\
      &       &       & \multicolumn{1}{l}{($A_2$: 587.7884)} &       &       &       &  \\\midrule
      &       &       &       &       &       &       &  \\
      \multicolumn{1}{l}{DB} &       &       & \multicolumn{1}{c}{Quadrature} & \multicolumn{4}{c}{Monte Carlo integration (Imp. Sampling)} \\
      \multicolumn{1}{c}{$T$} & $K$     &       & \multicolumn{1}{c}{Value} & \multicolumn{1}{c}{Value} & \multicolumn{1}{c}{Bias (\%)} & \multicolumn{1}{c}{Std. Error (\%)} &  \\\midrule
      \multicolumn{1}{l}{4 years} & 3     & \multicolumn{1}{l}{$A_{1,1}^1$} & \multicolumn{1}{c}{0.9866} & \multicolumn{1}{c}{0.9865} & \multicolumn{1}{c}{0.0047} & \multicolumn{1}{c}{0.0051} &  \\
      & $(j,i) = (1,1)$ & \multicolumn{1}{l}{$A_{1,1}^2$} & \multicolumn{1}{c}{0.1122} & \multicolumn{1}{c}{0.1114} & \multicolumn{1}{c}{0.6452} & \multicolumn{1}{c}{2.0849} &  \\
      &       & \multicolumn{1}{l}{CPU} & \multicolumn{1}{c}{9.0793} & 30.3230 & &         &  \\
      &       &       & \multicolumn{1}{l}{($A_{1,1}^1$: 0.4369)} &       &       &       &  \\
      &       &       & \multicolumn{1}{l}{($A_{1,1}^2$: 8.6424)} &       &       &       &  \\
      &       &       &       &       &       &       &  \\
      &       & \multicolumn{1}{l}{$A_{1,2}^1$} & \multicolumn{1}{c}{0.9866} & \multicolumn{1}{c}{0.9866} & \multicolumn{1}{c}{0.0040} & \multicolumn{1}{c}{0.0047} &  \\
      & $(j,i) = (1,2)$ & \multicolumn{1}{l}{$A_{1,2}^2$} & \multicolumn{1}{c}{0.1239} & \multicolumn{1}{c}{0.1247} & \multicolumn{1}{c}{0.6851} & \multicolumn{1}{c}{0.3349} &  \\
      &       & \multicolumn{1}{l}{CPU} & \multicolumn{1}{c}{9.0793} & 30.2637 &  &         &  \\
      &       &       & \multicolumn{1}{l}{($A_{1,2}^1$: 0.4369)} &       &       &       &  \\
      &       &       & \multicolumn{1}{l}{($A_{1,2}^2$: 8.6424)} &       &       &       &  \\
      &       &       &       &       &       &       &  \\
      &       & \multicolumn{1}{l}{$A_{2,1}^1$} & \multicolumn{1}{c}{0.9703} & \multicolumn{1}{c}{0.9704} & \multicolumn{1}{c}{0.0100} & \multicolumn{1}{c}{0.0074} &  \\
      &$(j,i) = (2,1)$ & \multicolumn{1}{l}{$A_{2,1}^2$} & \multicolumn{1}{c}{0.1349} & \multicolumn{1}{c}{0.1353} & \multicolumn{1}{c}{0.2873} & \multicolumn{1}{c}{0.2220} &  \\
      &       & \multicolumn{1}{l}{CPU} & \multicolumn{1}{c}{834.0646} & 52.1560 & &        &  \\
      &       &       & \multicolumn{1}{l}{($A_{2,1}^1$: 1.1743)} &       &       &       &  \\
      &       &       & \multicolumn{1}{l}{($A_{2,1}^2$: 832.8903)} &       &       &       &  \\
      &       &       &       &       &       &       &  \\
      &       & \multicolumn{1}{l}{$A_{2,2}^1$} & \multicolumn{1}{c}{0.9703} & \multicolumn{1}{c}{0.9705} & \multicolumn{1}{c}{0.0159} & \multicolumn{1}{c}{0.0079} &  \\
      & $(j,i) = (2,2)$ & \multicolumn{1}{l}{$A_{2,2}^2$} & \multicolumn{1}{c}{0.1464} & \multicolumn{1}{c}{0.1461} & \multicolumn{1}{c}{0.2016} & \multicolumn{1}{c}{0.1338} &  \\
      &       & \multicolumn{1}{l}{CPU} & \multicolumn{1}{c}{834.0646} & 52.0597 &  &        &  \\
      &       &       & \multicolumn{1}{l}{($A_{2,2}^1$: 1.1743)} &       &       &       &  \\
      &       &       & \multicolumn{1}{l}{($A_{2,2}^2$: 832.8903)} &       &       &       &  \\
      &       &       &       &       &       &       &  \\
      &       & \multicolumn{1}{l}{$A_{2,3}^1$} & \multicolumn{1}{c}{0.9703} & \multicolumn{1}{c}{0.9703} & \multicolumn{1}{c}{0.0080} & \multicolumn{1}{c}{0.0069} &  \\
      & $(j,i) = (2,3)$ & \multicolumn{1}{l}{$A_{2,3}^2$} & \multicolumn{1}{c}{0.1570} & \multicolumn{1}{c}{0.1569} & \multicolumn{1}{c}{0.0378} & \multicolumn{1}{c}{0.0685} &  \\
      &       & \multicolumn{1}{l}{CPU} & \multicolumn{1}{c}{834.0646} & 52.9188 & &        &  \\
      &       &       & \multicolumn{1}{l}{($A_{2,3}^1$: 1.1743)} &       &       &       &  \\
      &       &       & \multicolumn{1}{l}{($A_{2,3}^2$: 832.8903)} &       &       &       &  \\
      &       &       &       &       &       &       &  \\
      &       & \multicolumn{1}{l}{$A_{2,4}^1$} & \multicolumn{1}{c}{0.9703} & \multicolumn{1}{c}{0.9703} & \multicolumn{1}{c}{0.0012} & \multicolumn{1}{c}{0.0075} &  \\
      & $(j,i) = (2,4)$ & \multicolumn{1}{l}{$A_{2,4}^2$} & \multicolumn{1}{c}{0.1670} & \multicolumn{1}{c}{0.1669} & \multicolumn{1}{c}{0.0275} & \multicolumn{1}{c}{0.0588} &  \\
      &       & \multicolumn{1}{l}{CPU} & \multicolumn{1}{c}{834.0646} & 52.9617 & &      &  \\
      &       &       & \multicolumn{1}{l}{($A_{2,4}^1$: 1.1743)} &       &       &       &  \\
      &       &       & \multicolumn{1}{l}{($A_{2,4}^2$: 832.8903)} &       &       &       &  \\\midrule
      &       &       &       &       &       &       &  \\
      \multicolumn{1}{l}{SB} &       &       & \multicolumn{1}{c}{Quadrature} & \multicolumn{4}{c}{Monte Carlo integration (Imp. Sampling)} \\
      \multicolumn{1}{c}{$T$} & $K$     &       & \multicolumn{1}{c}{Value} & \multicolumn{1}{c}{Value} & \multicolumn{1}{c}{Bias (\%)} & \multicolumn{1}{c}{Std. Error (\%)} &  \\\midrule
      \multicolumn{1}{l}{4 years} & 3     & \multicolumn{1}{l}{$B_2^1$} & \multicolumn{1}{c}{0.9871} & \multicolumn{1}{c}{0.9871} & \multicolumn{1}{c}{0.0030} & \multicolumn{1}{c}{0.0029} &  \\
      & $(i = 2)$ & \multicolumn{1}{l}{$B_2^2$} & \multicolumn{1}{c}{0.9717} & \multicolumn{1}{c}{0.9717} & \multicolumn{1}{c}{0.0066} & \multicolumn{1}{c}{0.0041} &  \\
      &       & \multicolumn{1}{l}{CPU} & 1.1314 & 31.7341 & &     &  \\
      &       &       & \multicolumn{1}{l}{($B_2^1$:  0.1280)} &       &       &       &  \\
      &       &       & \multicolumn{1}{l}{($B_2^2$: 1.0034)} &       &       &       &  \\
      \bottomrule
  \end{tabular}}
  \label{tab:GMABbench_f}
\end{table}
 
 \end{appendix}

\end{document}